\documentclass[12pt,journal,letterpaper,onecolumn]{IEEEtran}

\usepackage{amssymb}
 \usepackage{amsmath}
 \usepackage{bbm}
\usepackage{graphicx}
\usepackage{subfigure}

\newtheorem{definition}{Definition}{\itshape}{\rmfamily}

\newcommand{\RR}{\mathbb{R}}

\newcommand{\DB}{\mathcal{D}}

\newcommand{\SSS}{\mathcal{S}}
\newcommand{\SSSo}{\ensuremath{\mathcal{S}^{o}}}
\newcommand{\SSSoo}{\ensuremath{\mathcal{S}^{o}\setminus\{o'\}}}
\newcommand{\SSSX}{\ensuremath{\mathcal{S}^{o_X}}}
\newcommand{\SSSY}{\ensuremath{\mathcal{S}^{o_Y}}}
\newcommand{\SSSXY}{\ensuremath{\mathcal{S}^{o_X}\setminus\{o_Y\}}}
\newcommand{\SSSYX}{\ensuremath{\mathcal{S}^{o_Y}\setminus\{o_X\}}}
\newcommand{\SSSCI}{\ensuremath{\mathcal{S}^{y}}}

\newcommand{\PRaUD}{\textbf{PSR}}
\newcommand{\Feifei}{\textbf{YLKS}}

\renewcommand{\P}{\ensuremath{P}} 
\newcommand{\p}{\ensuremath{\hat{P}}} 

\newtheorem{lemma}{Lemma}{\itshape}{\rmfamily}
{\itshape}{\rmfamily}

\begin{document}

\title{Scalable Probabilistic Similarity Ranking in Uncertain Databases (Technical Report)}

\author{Thomas Bernecker,
        Hans-Peter Kriegel,
        Nikos Mamoulis,
        Matthias Renz,
        and~Andreas~Zuefle
\thanks{}
\thanks{T. Bernecker, H.P. Kriegel, M. Renz, and A. Zuefle are with the Ludwig-Maximilians-Universit\"{a}t,
M\"{u}nchen, Germany. E-mail:
\{bernecker,kriegel,renz,zuefle\}@dbs.ifi.lmu.de. N. Mamoulis is
with the University of Hong Kong, Pokfulam Road, Hong Kong.
E-mail: nikos@cs.hku.hk.} }

\maketitle

\begin{abstract}
This paper introduces a scalable approach for probabilistic top-k
similarity ranking on uncertain vector data. Each uncertain object
is represented by a set of vector instances that are assumed to be
mutually-exclusive. The objective is to rank the uncertain data
according to their distance to a reference object. We propose a
framework that incrementally computes for each object instance and
ranking position, the probability of the object falling at that
ranking position. The resulting rank probability distribution can
serve as input for several state-of-the-art probabilistic ranking
models. Existing approaches compute this probability distribution
by applying a dynamic programming approach of quadratic
complexity. In this paper we theoretically as well as
experimentally show that our framework reduces this to a
linear-time complexity while having the same memory requirements,
facilitated by  incremental accessing of the uncertain vector
instances in increasing order of their distance to the reference
object. Furthermore, we show how the output of our method can be
used to apply probabilistic top-k ranking for the objects,
according to different state-of-the-art definitions. We conduct an
experimental evaluation on synthetic and real data, which
demonstrates the efficiency of our approach.
\end{abstract}

\section{Introduction}
\noindent In the past two decades, there has been a great deal of
interest in developing efficient and effective methods for
similarity queries in spatial, temporal, multimedia and sensor
databases. Similarity ranking is a hot topic in database research
because a large number of emerging applications require
exploratory querying on the aforementioned databases. A ranking
query orders the objects in a database with respect to their
similarity to a reference object. In a spatial database context,
nearest neighbor queries rank the contents of a spatial object set
(e.g., restaurants) in increasing order of their distance to a
reference location. In a database of images, a similarity query
ranks the feature vectors of images in increasing order of their
distance (i.e., dissimilarity) to a query image.

More recently, it has been recognized that many applications
dealing with spatial, temporal, multimedia, and sensor data have
to cope with uncertain or imprecise data. For instance, in the
spatial domain, the locations of objects usually change
continuously, thus the positions tracked by GPS devices are often
imprecise. Similarly, vectors of values collected in sensor
networks (e.g., temperature, humidity, etc.) are usually
inaccurate, due to errors in the sensing devices or time delays in
the transmission. Finally, images collected by cameras may have
errors, due to low-resolution or noise. As a consequence, there is
a need to adapt storage models and indexing/search techniques to
deal with uncertainty. There is already a volume of research on
probabilistic data models \cite{bshw-06,rds-07,sd-07,ajko-07}.

In this paper, we focus on similarity ranking of uncertain vector
data. Prior work in this direction includes
\cite{SS-CKP-03,SS-CXPSV-04,SS-TCXNKP-05,bps-06,uncertain-kkpr-06,uncertain-kkr-07,cly-09,si-09}.
In a nutshell, there are two models for capturing uncertainty of
objects in a high dimensional space. In the {\em continuous}
uncertainty model, the uncertain values of an object are
represented by a continuous probability distribution function
(pdf) within the vector space. This type of representation is
often used in applications where the uncertain values are assumed
to follow a specific probability density function (pdf), e.g. a
Gaussian distribution \cite{bps-06}. Similarity search methods
based on this model involve expensive integrations of the pdf's,
thus special approximation techniques for efficient query
processing are typically employed \cite{SS-TCXNKP-05}. In the {\em
discrete} uncertainty model, each object is represented by a
discrete set of alternative values, and each value is associated
with a probability \cite{uncertain-kkpr-06}. The main motivation
of this representation is that, in most real applications, data
are collected in a discrete form (e.g., information derived from
sensor devices). In this paper, we adopt the discrete uncertainty
model which also complies with the x-relations model used in the
\emph{Trio} system \cite{abshnsw-06}.

Consider, for example, a set of three two-dimensional objects $A$,
$B$, and $C$ (e.g., locations of mobile users), and their
corresponding uncertain instances $\{a_1,a_2\}$,
$\{b_1,b_2,b_3\}$, and $\{c_1,c_2,c_3\}$, as shown in Figure
\ref{fig:ex_simple}(a). Each instance carries a probability (shown
in brackets) and instances of the same object are
mutually-exclusive. In addition, the sum of probabilities of each
object's instances cannot exceed 1. Assume that we wish to rank
the objects $A$, $B$, and $C$ according to their distances to the
query point $q$ shown in the figure. Clearly, several rankings are
possible. In specific, each combination of object instances
defines an order. For example, for combination $\{a_1,b_1,c_1\}$
the object ranking is $(B,A,C)$ while for combination
$\{a_2,b_3,c_1\}$ the object ranking is $(A,B,C)$. Each
combination corresponds to a {\em possible world}
\cite{abshnsw-06}, whose probability can be computed by
multiplying the probabilities of the instances that comprise it,
assuming independent existence probabilities between the instances
of different objects.

\begin{figure*}[hbt]
\centering \subfigure[Object Instances]{
       \includegraphics[width=0.5\columnwidth]{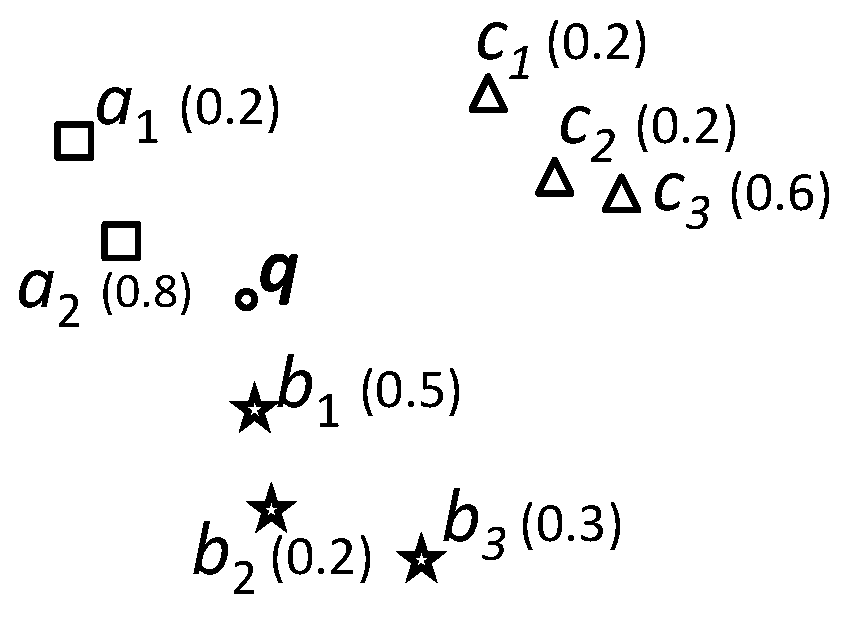}
        }\hspace{2cm}
\subfigure[Bipartite Graph]{
       \includegraphics[width=0.25\columnwidth]{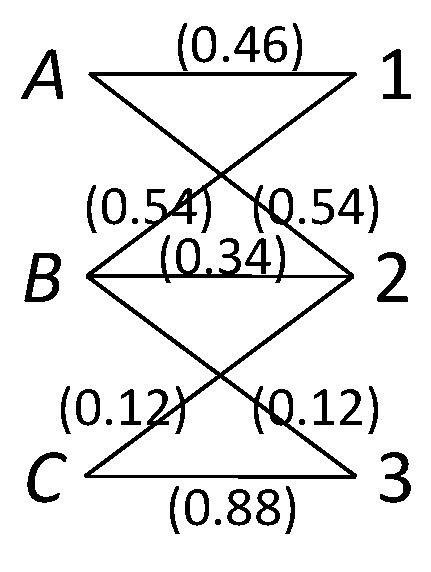}
        }
\caption{Object Instances and Rank Probability
Graph}\label{fig:ex_simple}
\end{figure*}

The example illustrates the ambiguity of ranking in uncertain
data. On the other hand, most applications require the definition
of a non-ambigous object ranking. For example, assume that a
robbery took place at location $q$ and the objects correspond to
the positions of suspects that are sampled around the time that
the robbery took place. The probabilities of the samples depend on
various factors (e.g., time-difference of the sample to the
robbery event, errors of capturing devices, etc.). As an
application, we may want to define a definite probabilistic
proximity ordering of the suspects to the event, in order to
prioritize interrogations.

Various top-$k$ query approaches have been proposed generating
un-ambiguous rankings from probabilistic data. Examples are
U-top$k$ \cite{uncertain-sic-07}, U-$k$Ranks
\cite{uncertain-sic-07}, PT-$k$ \cite{hpzl-08}, Global top-$k$
\cite{zc-08}, and expected rank \cite{cly-09}. A summary of these
ranking models can be found in \cite{cly-09}. All of them attempt
to weigh the objects based on their probability to be in each of
the first $k$ ranks, but they use different ways to define the
weights.

A common module in most of these approaches is the computation for
each object instance $x$ the probability $P_i$ that $i$ objects
are closer to $q$ than $x$ for all $1\leq i\leq k$. The resulting
probabilities are aggregated to build the probability of each
object at each rank. For example, the U-$k$Ranks query reports the
$i^{th}$ result as the object that is the most likely to be ranked
$i^{th}$ over all possible worlds. For this computation, we
obviously need the probabilities of all instances to be ranked
$i^{th}$ over all possible worlds. The probability that an object
is ranked at a specific position $i$ can be computed by summing
the probabilities of the possible worlds that support this
occurrence. In our example, the probability that object $A$ occurs
as first one is 0.46 and the probability that object $B$ is the
first is 0.54. All possible occurrences and the corresponding
probabilities are represented by the object-rank bipartite graph
which is shown in Figure \ref{fig:ex_simple}(b). Non-existing
edges imply zero probability, i.e. it is not possible that the
object occurs at the corresponding ranking position. In this
example, all instances of $A$ precede all those of $C$, so $C$
cannot occur as first object and $A$ cannot be ranked to the last
position.

In this paper, we propose a framework that, given a database with
uncertain vector objects, computes the rank probabilities of the
object instances (e.g., $a_1$) in linear time to the total number
of instances of all objects.
 --- assuming that
the instances are accessed in increasing distance order to the
query object $q$ (e.g., with the help of a nearest neighbor search
algorithm \cite{hs-95}). As these can be aggregated on-the-fly,
our framework also computes the rank probabilities of the objects
(e.g., $A$) at the same cost. This is a great improvement, over
the state-of-the-art \cite{tkde-ylks-08}, which computes these
probabilities in quadratic time.

\subsection{Problem Definition}
\label{sec:problemDefinition}

Analogously to the \emph{Trio} \cite{abshnsw-06} system, we define
an uncertain database as a set of uncertain objects (x-tuples),
each including a number of alternatives associated with
probabilities. Here, we consider uncertain vector objects in a
$d$-dimensional vector space, i.e., each object is assigned to
multiple alternative positions associated with a probability
value. Let us note that this model assumes independence among the
uncertain objects.

\begin{definition}[Uncertain Vector Objects]
\label{def:uncertainVectorObjects} An uncertain vector object $o$
corresponds to a finite set of alternative points in a
$d$-dimensional vector space, called \emph{object instances}, each
associated with a probability value, i.e., $o = \{(x,p)$, where
$x\in\RR^d$, and $p\in[0,1]\}$ is the probability that $o$ has
position $x$. The probabilities of the object instances represent
a discrete probability distribution of the alternative points,
such that the condition $\sum_{(x,p)\in o}p \leq 1$ holds. The
collection of instances of all objects forms the uncertain
database $\DB$.
\end{definition}

Note that the condition $\sum_{(x,p)\in o}p < 1$ implies
existential uncertainty, meaning that the object may not exist at
all. We assume that the database objects are already given in the
discrete representation as specified above. In case of an
uncertain database where the uncertain objects are represented by
a continuous probability function (pdf), the generally applicable
concept of sampling can be used to transform the objects to the
discrete representation as defined above.

Given a database of uncertain vector objects, our goal is to
compute for each object instance and for the first $k$ rank
positions, the probability of the object to be in that position.

\begin{definition}[Rank Probability]
\label{def:probabilisticRanking} Given a query point $q$, an
instance $x\in\DB$, belonging to object $o$, and a rank
$i\in\{1,\dots,k\}$, the \emph{probabilistic rank} $p\_rank_q:
(\DB\times\{1,\dots,k\})\rightarrow [0,1]$ reports how likely
$(i-1)$ objects $o'\ne o$ are closer to $q$ than $x$, i.e., the
probability that $x$ is at the $i^{th}$ ranking position according
to the distance (i.e., dissimilarity) between $x$ and $q$.
\end{definition}

Since the number of possible worlds is exponential in the number
of uncertain objects, it is impractical to enumerate all of them
in order to find the rank probabilities of all object instances.
Recently, it has been shown in \cite{uncertain-ylks-08} that we
can compute the probabilities between all object instances and
ranks in $O(kn^2)$ time, where $n$ is the number of object
instances required to be accessed until the solution is confirmed.
This solution can be applied to all problems that comply to the
x-relation model (including our problem). In this paper, we
propose a significant improvement of this approach, which reduces
the time complexity to $O(kn)$.

In Section \ref{sec:ProbabilisticRankingUncertainObjects}, we
discuss in detail how our method can be used as a module in
various models that rank the objects according to the rank
probabilities of their instances.

Although in the paper, we focus on databases of uncertain objects
as in Definition \ref{def:uncertainVectorObjects}, our results
apply in general to x-relations as defined in \cite{abshnsw-06},
which model mutual-exclusiveness constraints between existentially
uncertain tuples (i.e., object instances in our model). Thus, our
method is general and it can be used irrespectively to whether we
have uncertain objects or existentially uncertain tuples with
exclusiveness constraints, expressed by x-tuples.

\subsection{Contributions and Outline}

The main contributions of this paper can be summarized as follows:
\begin{itemize}

    \item 
    We propose a framework
    based on iterative distance browsing that efficiently supports probabilistic similarity ranking in uncertain vector databases.

    \item 
    We present a novel and theoretically founded
    approach for computing the rank probabilities of each object. We prove that our method reduces the
    computational cost of the rank probabilities from O($kn^2$), achieved by
    the best currently known method, to O($kn$).

    \item 
    We show
    how diverse state-of-the-art probabilistic ranking models can use our framework to accelerate computation.

    \item 
    We conduct an experimental evaluation, using real and synthetic data, which demonstrates the applicability of our framework and verifies
our theoretical findings.

\end{itemize}

The rest of the paper is organized as follows: In the next
section, we survey existing work in the field of managing and
querying uncertain data. In Section
\ref{sec:probsimrankingmethods}, we introduce our framework for
computing the rank probabilities of uncertain object instances,
followed by the details regarding the efficient incremental rank
probability computation for each object instance.

The complete algorithm for computing the rank probabilities for
all instances and the corresponding objects is presented in
Section \ref{sec:Algorithm}. We experimentally evaluate the
efficiency of our approach in Section \ref{sec:Experiments} and
conclude the paper in Section \ref{sec:conclusions}.

\section{Related Work}
\label{sec:relatedWork}

The potential of uncertain data processing has achieved increasing
interest in diverse application fields, e.g., sensor monitoring
\cite{csp-05}, traffic analysis and location-based services
\cite{wscy-99}, etc.

By now, uncertain data management has been established as an
important branch of research within the database community, with
increasing tendency. Existing approaches in this field of
modelling of, managing of and query processing on uncertain data
can be categorized into diverse directions, including
probabilistic databases \cite{bshw-06,rds-07,sd-07,ajko-07},
indexing of uncertain data \cite{SS-CXPSV-04,SS-TCXNKP-05,bps-06}
and probabilistic query processing
\cite{SS-CKP-03,ds-07,bps-06,uncertain-kkpr-06,bkr-08,tkde-ylks-08,uncertain-sic-07}.

Probabilistic databases usually relate to probabilistic relational
data, i.e. relations with uncertain tuples \cite{ds-07}, and use
the possible worlds semantic \cite{ajko-07} which is a probability
distribution on all possible database instances; a database
instance corresponds to a subset of uncertain tuples. In the
general model, the possible worlds are constrained by rules that
are defined on the tuples in order to incorporate object (tuple)
correlations \cite{sd-07}. The ULDB model proposed in
\cite{bshw-06} and used in the \emph{Trio} \cite{abshnsw-06}
system supports uncertain tuples with alternative instances which
are called x-tuples. Relations in ULDB are called x-relations
containing a set of x-tuples. Each x-tuple corresponds to a set of
tuple instances which are assumed to be mutually exclusive, i.e.
no more than one instance of an x-tuple can appear in a possible
world instance at the same time. This probabilistic data model
closes the gap between two prevalent uncertainty models, the
\emph{tuple uncertainty} \cite{ds-07} and the \emph{attribute
uncertainty} \cite{SS-CKP-03}. An x-tuple is able to model an
object with attribute value uncertainty; i.e., the instances of an
x-tuple represent the probability value distribution of the
corresponding uncertain attribute.

In this paper, we adopt this concept to model uncertain vector
objects. An uncertain vector object would correspond to an x-tuple
of alternative uncertain instances of the object. Several
approaches for indexing uncertain vector objects have been
proposed \cite{SS-CXPSV-04,SS-TCXNKP-05,bps-06,YiuMDTV09}. They
mainly differ in the uncertainty model supported and in the type
of supported similarity queries. In \cite{bps-06}, the Gauss-tree
is introduced, which is an index for managing large amounts of
uncertain objects with their uncertain attribute represented by a
Gaussian distribution function. The proposed system aims at
efficiently answering \emph{identification queries} like ``Give me
all persons in the database that could be shown on a given image
with a probability of at least 10\%''. Additionally, \cite{bps-06}
proposed \emph{probabilistic identification queries} which are
based on a Bayesian setting. Later, in \cite{bps-06} an approach
for incrementally retrieving the $k$ most likely uncertain objects
that might be placed in a given query interval is proposed. Note
that this definition is sematically different than the problem
studied in this paper.

In \cite{bps-06}, objects which have the highest probability of
being located inside a given query range are reported. In
contrast, the approaches for managing uncertain vector objects
proposed in \cite{SS-CKP-03,SS-CXPSV-04,SS-TCXNKP-05} support
arbitrarily shaped probability distribution functions for
uncertain object attributes. Similar to \cite{bps-06}, the
approaches in \cite{SS-CXPSV-04,SS-TCXNKP-05} focus on probability
computations based on query predicates according to a given query
range and, thus, are not applicable for our problem. Although
\cite{YiuMDTV09} studies probabilistic ranking of objects
according to their distance from a reference query point, the
solutions are limited to existentially uncertain spatial data with
a single alternative.

We can categorize existing probabilistic querying approaches
according to the uncertainty model they use. While probabilistic
similarity queries over uncertain vector data are dedicated to the
attribute value uncertainty model
\cite{SS-CKP-03,uncertain-kkpr-06}, probabilistic top-$k$ query
approaches are usually associated with tuple uncertainty data in
probabilistic databases
\cite{uncertain-sic-07,uncertain-ylks-08,rds-07,tkde-ylks-08}.
There exists a third probabilistic query category concerning
spatially extended uncertain data as proposed in \cite{ls-08}. But
there is only little work in this direction.

To the best of our knowledge, only \cite{bkr-08} addresses
probabilistic ranking according to our problem definition. There,
a divide and conquer method for accelerating the computation of
the ranking probabilities is proposed. Although the proposed
approach achieves a significant speed-up compared to the naive
solution incorporating each possible database instance, its
runtime is still exponential. Related to our ranking problem,
significant work has been done in the field of probabilistic
top-$k$ query processing. Soliman et al. \cite{uncertain-sic-07}
were the first who studied such problems on the x-relations model
of \cite{bshw-06}. They proposed two ways of ranking uncertain
tuples. In the first, \emph{uncertain top-$k$} (U-Top$k$) query,
the objective is to find the $k$-permutation of the most likely
tuples to be the top-$k$. In our setting, this corresponds to
finding the top-$k$ most probable object instances (belonging to
different objects) in all possible worlds. The \emph{uncertain
$k$-ranks query} (U-$k$Ranks) reports a probabilistic ranking of
the tuples (again, {\em not} the x-tuples). However, an efficient
approach for this problem is only given for the case where the
tuples are mutually independent which does not hold for the
x-relation model. At the same time Re et al. proposed in
\cite{rds-07} an efficient but approximative probabilistic ranking
based on the concept of Monte-Carlo simulation. Later, Yi et al.
proposed in
 \cite{tkde-ylks-08}
the first efficient exact probabilistic ranking approach for the
x-relation model, for both cases of single-alternative x-tuples
only, i.e. x-tuples with only one uncertain instance, and
multi-alternative x-tuples. They proposed dynamic programming
based methods for the computation of uncertain ranking queries,
which have much lower costs than the previously best known
results. Furthermore, they proposed early stopping conditions for
accessing the tuples. Their methods for U-Top$k$ and U-$k$Ranks
queries have O($nlogk$) and O($kn^2$) time complexity,
respectively. The cost of the U-$k$Ranks algorithm is dominated by
the computation of the probability of each accessed tuple to be in
each of the $k$ first ranks. In this paper, we also use this as a
module of finding the object-rank probabilities. However, we
propose an improvement of their  O($kn^2$) algorithm that does the
same work in O($kn$) without increasing the memory requirements.

In a recent paper, Cormode et al. \cite{cly-09} reviewed
alternative top-$k$ ranking approaches for uncertain data,
including the U-Top$k$ and U-$k$Ranks queries, and argued for a
more robust definition of ranking, namely the {\em expected} rank
for each tuple (or x-tuple). This is defined by the weighted sum
of the ranks of the tuple in all possible worlds, where each world
in the sum is weighed by its probability. The $k$ tuples with the
lowest expected ranks are argued to be a more appropriate
definition of a top-$k$ query than previous approaches.
Nevertheless, we found by experimentation that such a definition
may not be appropriate for ranking objects (i.e., x-tuples), whose
instances have large variance (i.e., they are scattered far from
each other in space). In general, the result of this ranking
method is similar to the brute-force approach that would take the
mean of the instances for each object and rank these means. On the
other hand, approaches that take into consideration the rank
probabilities (e.g., U-$k$Ranks) would be more suitable for such
data. This is the reason why we focus on the computation of rank
probabilities in this paper. Another piece of recent related work
is \cite{si-09}, where the goal is to rank uncertain objects
(i.e., x-tuples) whose score is uncertain and can be described by
a range of values. Based on these ranges, the authors define a
graph that captures the partial orders among objects. This graph
is then processed to compute U-$k$Ranks and other queries.
Although this work has similar objectives to ours, it operates on
a different input, where the distribution of uncertain scores is
already known, as opposed to our work which dynamically computes
this distribution by performing a linear scan over the ordered
object instances.


\section{Probabilistic Ranking Framework}
\label{sec:probsimrankingmethods} \noindent Our framework
basically consists of two modules which are performed in an
iterative way:
\begin{itemize}
    \item The first module ({\em distance browsing}) incrementally retrieves the instances of all
objects in order of their distance to $q$. This can be achieved
with the help of a multi-dimensional index (e.g., an $R^*$-tree
index \cite{kssb-90}), using an incremental nearest neighbor
search algorithm \cite{hs-95}.

    \item The second module computes the probabilistic ranks $p\_rank_q(o_j,i)$ of each object
    instance $o_j$ reported from the distance browsing for all $1\leq i\leq
    k$. This step is the main focus of this
paper, because of its potentially high computational cost. A naive
solution would take into account all possible worlds that include
the instance and update the probabilities accordingly, however, as
discussed before, there already exists an efficient solution which
can perform this computation in quadratic time and linear space
\cite{uncertain-ylks-08}. In this paper, we improve this method to
a linear time and space complexity algorithm. The key idea is to
use the probabilistic ranks of the previous object instance to
derive those of the currently accessed one in O($k$) time. Section
\ref{subsec:IncrementalProbabilityComputation} has the details of
this improvement.
\end{itemize}

Our framework is illustrated in Figure
\ref{fig:ProbabilisticRankingSteps}. The computation of the
probability distributions is iteratively processed within a loop.
First, we initialize a distance browsing among the object
instances starting from a given query point $q$. For each object
instance fetched from the distance browsing (Module 1), we compute
the corresponding rank probabilities (Module 2) and update the
rank probability distributions generated from the probabilistic
ranking routine.

Note that the rank probabilities of the object instances (i.e.,
tuples in the x-relations model) reported from the second module
can be optionally aggregated into rank probabilities of the
objects (i.e., x-tuples in the x-relations model). The probability
that an uncertain vector object $o=\{(x_1,p_1),\ldots,(x_s,p_s)\}$
is at the $i^{th}$ ranking position according to the distance
between $o$ and a reference query object $q$ is
$$
p\_rank_q(o,i) = \sum_{l=1,\ldots,s}p_l\cdot p\_rank_q(x_l,i).
$$
Our framework can be used to compute the object-based rank
probabilities by maintaining a list of objects from which
instances have been seen so far and successively aggregate the
rank probabilities by means of the instance-based rank
probabilities reported from the framework.

Finally, in a postprocessing step, the rank probability
distributions computed by our framework can be used to generate a
definite ranking of the objects or object instances. The objective
is to find a non-ambiguous ranking where each object or object
instance is uniquely assigned to one rank. Here, one can plug-in
any user-defined ranking method that requires rank probability
distributions of objects in order to compute unique positions. In
Section \ref{sec:ProbabilisticRankingUncertainObjects}, we
illustrate this for several well-known probabilistic ranking
queries that make use of such distributions. In particular, we
demonstrate that by using our framework we can process such
queries in O($n log n + k\cdot n$) time\footnote{Note that the
O($n log n$) factor is due to pre-sorting the object instances
according to their distances to the query object. If we assume
that the instances are already sorted then our framework can
compute the probability distributions for the first $k$ rank
positions in O($k\cdot n$) time.}, as opposed to existing
approaches that require O($k\cdot n^2$) time.

\begin{figure*}[t]
        \centering
    \hspace{\fill}
        \includegraphics[width=0.8\columnwidth]{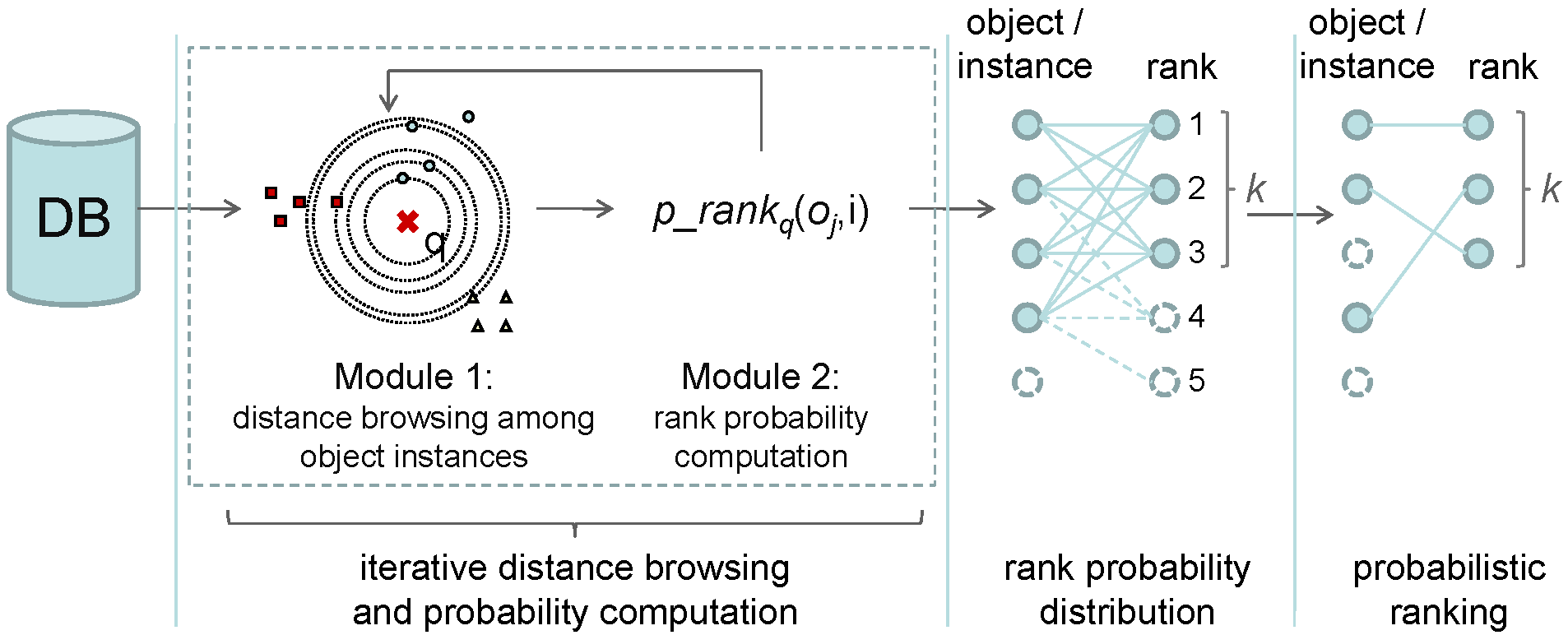}
    \hspace{\fill}
    \caption{Framework for probabilistic similarity ranking.}
    \label{fig:ProbabilisticRankingSteps}
\end{figure*}

\begin{table}[hbt] \centering
\begin{tabular}{|p{8mm}|p{60mm}|}
  \hline
  \multicolumn{2}{|c|}{Table of Notations} \\
  \hline
  $\DB$ & an uncertain database\\
  $N$ & the cardinality of $\DB$ \\
  $q$ & a query vector in respect to which a probabilistic ranking is
  computed\\
  $k$ & the ranking depth that determines the number of ranking positions of the
  ranking query result \\ $D$ & a distance browsing of $\DB$ with respect to
  $q$ \\ $o$ & an uncertain vector object corresponding to a finite set of
  alternative vector point instances\\
  $x$,$y$ & vector point instances\\
  $o_X,o_Y$ & the objects that the instances $x$ and $y$ respectively belong to\\
  $p$ & the probability that an uncertain vector object matches a given  vector position \\
  $\P_i(o)$ & the probability that object $o$ is assigned to the $i$-th ranking position $i$, i.e. the probability that
  exactly ($i$-1) objects in $(\DB \setminus \{o\})$ are closer to $q$ than $o$ \\
$\p_i(x)$ & the probability that an instance $x$ of object $o$ is
assigned to the $i$-th ranking position $i$, i.e.
 the probability that exactly $i-1$ objects in $(\DB \setminus \{o\})$ are
 closer to $q$ than $x$ \\ $AOL$ & \emph{Active Object List}\\
 $\SSS$ & a set of objects that have already been seen, i.e. the set that
 contains an object $o$ \emph{iff} at least one instance of $o$ has already been returned by the distance browsing
 $D$ \\
 $\SSSo$ & a set that contains the objects that have already been seen,
 except for object $o$, i.e. $\SSSo = \SSS \setminus \{o\}$ \\
 $P_{i,S,x}$ & the probability that exactly $i$
 objects $o$ $\in S$ are closer to $q$ than an object instance $x$\\
 $P_{x}(o)$ & the probability that object $o$ is closer to query point $q$
 than the vector point $x$; computable using Lemma \ref{lem:totalprobs}\\
  \hline
\end{tabular}
\caption{Table of notations used in this work.}
\label{table:notations}
\end{table}

\subsection{Dynamic Probability Computation}
\label{subsec:DynamicProbabilityComputation}

Consider an uncertain object $o$, defined by $m$ probabilistic
instances $o=\{(x_1,p_1),\dots,(x_m,p_m)\}$. The probability that
$o$ is assigned to a given ranking position $i$ is equal to the
chance that exactly $i-1$ objects $o' \in (\DB \setminus o)$ are
closer to the query object $q$ than the object $o$. This  can be
computed by aggregating the probabilities over all instances
$(x,p)$ of $o$ that exactly $i-1$ objects $o'$ are closer to $q$
than the instance $(x,p)$. Formally,

$$
\P_i(o) = \sum_{(x,p)\in o}\p_i(x|(o=x))
$$

where $\P_i(o)$ denotes the probability that o is assigned to the
ranking position $i$, i.e., exactly $i-1$ objects in $(\DB
\setminus o)$ are closer to $q$ than $o$. The conditional
probability $\p_i(x|(o=x))$ denotes the probability that exactly
$i-1$ objects in $(\DB \setminus o)$ are closer to $q$ than the
object instance $x$ of $o$, given that the object $o$ is in fact
at the location instance $x$. Since the conditional probability
$\p_i(x|(o=x))$ only depends on the objects $o' \in (\DB \setminus
o)$ and is thus independent of $o$, we obtain:
\begin{equation}
\label{equ:sumOfInstanceProbs} \P_i(o) = \sum_{(x,p)\in
o}(\p_i(x)\cdot \P(o=x)) = \sum_{(x,p)\in o}(\p_i(x)\cdot p)
\end{equation}

Based on the above formula we can compute the probabilities for an
object $o$ to be assigned to each of the ranking positions
$i\in\{1,\dots,k\}$ by computing the probabilities $P_i(x)$ for
all instances $(x,p)$ of $o$. As mentioned above, we perform this
computation in an iterative way, i.e., whenever we fetch a new
object instance $(x,p)$
we compute all probabilities $\p_i(x)\cdot p$ for all
$i\in\{1,\ldots,k\}$. Thereby, in a list we store the current {\em
probability state} according to all ranking positions
$i\in\{1,\ldots,k\}$ for each object for which we already have
accessed some instances and for which we expect to obtain further
instances in the remaining iterations. Whenever the probabilities
according to a new object instance are computed, we update the
list by adding the new probabilities to the current probability
state.

In the following, we show how to compute the probabilities
$\p_i(x)\cdot p$ for all $i\in\{1,\ldots,k\}$ for a given object
instance $(x,p)$ of an uncertain object $o$ which is assumed to be
currently fetched from the distance browsing (Step 1). For this
computation we first need, for all uncertain objects $o'\in\DB$,
the probability $P_x(o')$ that $o'$ is closer to $q$ than the
current object instance $x$. These probabilities are stored in an
\emph{Active Object List} $AOL$, which can easily be kept updated
due to the following obvious lemma:

\begin{lemma}
\label{lem:totalprobs} Let $q$ be the query object and $(x,p)$ be
the object instance of an object $o$ fetched from the distance
browsing in the current processing iteration. The probability that
an object $o'\neq o$ is closer to $q$ than $x$ is
$$
P_x(o')=\sum_{(x',p')\in o'} p'\mathrm{,}
$$
where (x',p') are the instances fetched in previous processing
iterations.
\end{lemma}

Lemma \ref{lem:totalprobs} says that we can accumulate in overall
linear space the sums of probabilities of all instances for each
object, which have been seen so far and use them to compute
$P_x(o')$ given the current instance $x$ and any object $o'$ in
$\mathcal{D}$. In fact, we only need to manage in the list the
probabilities of those objects for which we already have accessed
an instance and for which we expect to access further instances in
the remaining iterations.

Now let us see how we can use list $AOL$ to efficiently compute
the probabilities $\p_i(x)$. Assume that $(x,p)\in o$ is the
current object instance reported from distance browsing. Let
$\SSS=\{o_1,\ldots,o_j\}$ be the set of objects which have been
seen so far, i.e. for which we already have seen at least one
object instance. We use the same observation as in
\cite{uncertain-ylks-08}. The probability that an object
$o\in\SSS$ appears at ranking position $i$ of the first $j$
objects seen so far only depends on the event that $i-1$ of the
remaining $j-1$ objects $p\in\SSS$ ($p\neq o$) appear before $o$,
no matter which of these objects fulfill this criterion. Let
$\SSSo$ denote the set of objects seen so far without object $o$,
i.e. $\SSSo=\SSS\setminus\{o\}$. Furthermore, let $P_{i,\SSSo,x}$
denote the probability that exactly $i$ objects of $\SSSo$ are
closer to $q$ than the object instance $x$. Now, we can formulate
the recursive function:
$$
P_{i,\SSSo,x}=P_{i-1,\SSSoo,x}\cdot P_x(o') +
P_{i,\SSSoo,x}\cdot(1-P_x(o'))\mathrm{,}
$$
where
\begin{equation}
\label{equ:dynamicProgramming} P_{0,\emptyset,x} = 1\mbox{ and }
P_{i,\SSSo,x} = 0\mbox{, iff } i>|\SSSo|\vee i<0.
\end{equation}
The correctness of Equation \ref{equ:dynamicProgramming} can be
shown by the following intuition: the event that $i$ objects of
$\SSSo$ are closer to $q$ than $x$ occurs if one of the following
conditions holds. In the case that an object $o'\in\SSSo$ is
closer to $q$ than $x$, then $i-1$ objects of $\SSSoo$ must be
closer to $q$. Otherwise, if we assume that object $o'\in\SSSo$ is
farther to $q$ than $x$, then $i$ objects of $\SSSoo$ must be
closer to $q$.

For each object instance $(x,p)$ reported from the distance
browsing, we have to apply the recursive function as defined
above.

Specifically, we have to compute for each instance $(x,p)$ the
probabilities $P_{i,\SSSo,x}$ for all
$i\in\{0,\ldots,\min\{k,|\SSSo|\}\}$ and for $j=|\SSSo|$ subsets
of $\SSSo$. If $n=|\DB|$, this has a cost factor of O($k\cdot n$)
per object instance retrieved from the distance browsing, leading
to a total cost of O($k\cdot n^2$). Assuming that $k$ is a small
constant, we have an overall runtime of O($n^2$).

In the following, we show how we can compute each $P_{i,\SSSo,x}$
in constant time by utilizing the probabilities computed for the
previously accessed instance.

\subsection{Incremental Probability Computation}
\label{subsec:IncrementalProbabilityComputation}

Let $(x,p_x)\in o_X$ and $(y,p_y)\in o_Y$ be two object instances
consecutively returned from the distance browsing. W.l.o.g. let
$(x,p_x)$ be returned before $(y,p_y)$. Each of the probabilities
$P_{i,\SSSY,y}$ ($i\in\{0,\ldots,|\SSSY|\}$) can be computed from
the probabilities $P_{i,\SSSX,x}$ in constant time. In fact, the
probabilities $P_{i,\SSSY,y}$ can be computed by considering at
most one recursion step backward.

The following three cases have to be considered. The first two are
easy to tackle and the third one is the most common and
challenging one.
\begin{description}
    \item \textbf{Case 1:} Both instances belong to the same object, i.e.
    $o_X = o_Y$.
    \item \textbf{Case 2:} Both instances belong to different objects,
    i.e. $o_X\neq o_Y$ and $(y,p_y)$ is the first
    returned instance of object $o_Y$.
    \item \textbf{Case 3:} Both instances belong to different objects,
    i.e. $o_X\neq o_Y$ and $(y,p_y)$ is not the first
    returned instance of object $o_Y$.
\end{description}

\begin{figure*}[t]
        \centering
    \subfigure[Case 1: Instances $(x,p_x)$ and $(y,p_y)$ belong to the same
    object.\label{fig:threeCases_case1}]{
        \includegraphics[width=0.29\columnwidth]{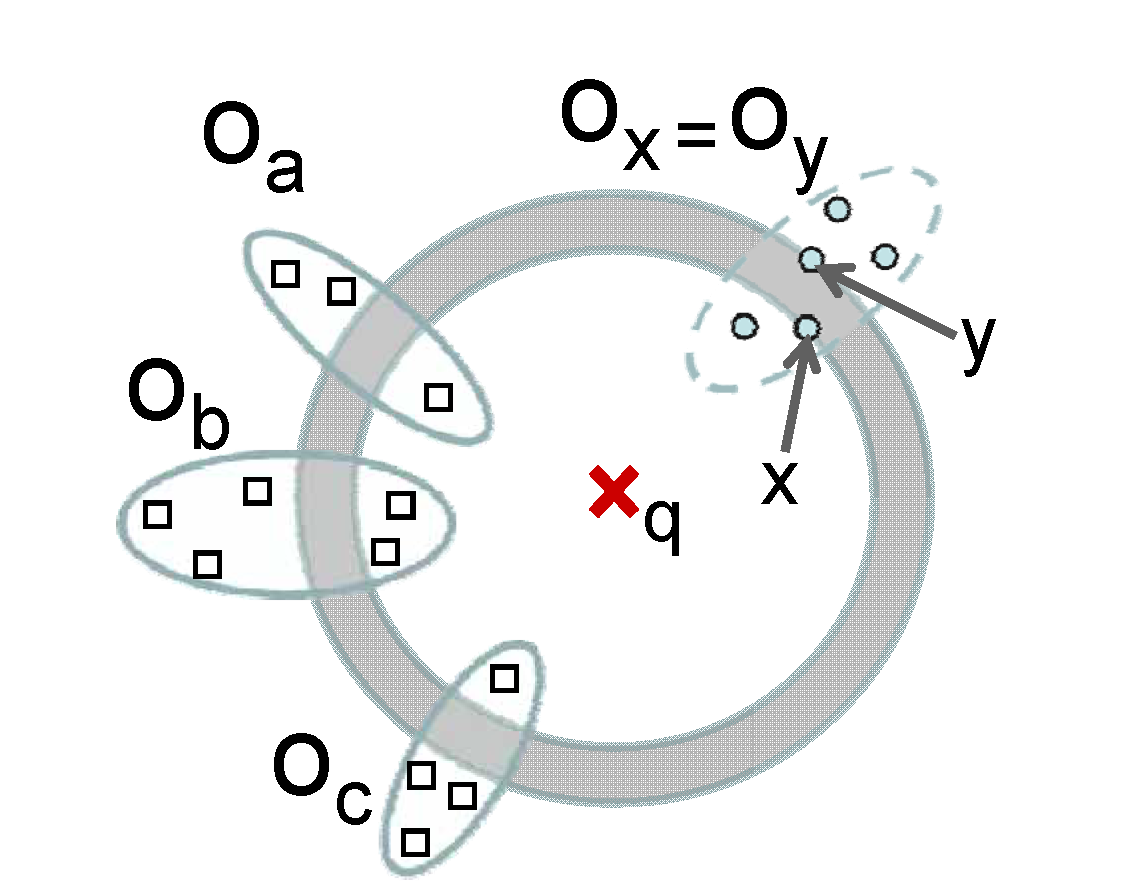}
        }
    \subfigure[Case 2: Instance $(y,p_y)$ is the first
    returned instance of object $o_Y$.\label{fig:threeCases_case2}]{
        \includegraphics[width=0.29\columnwidth]{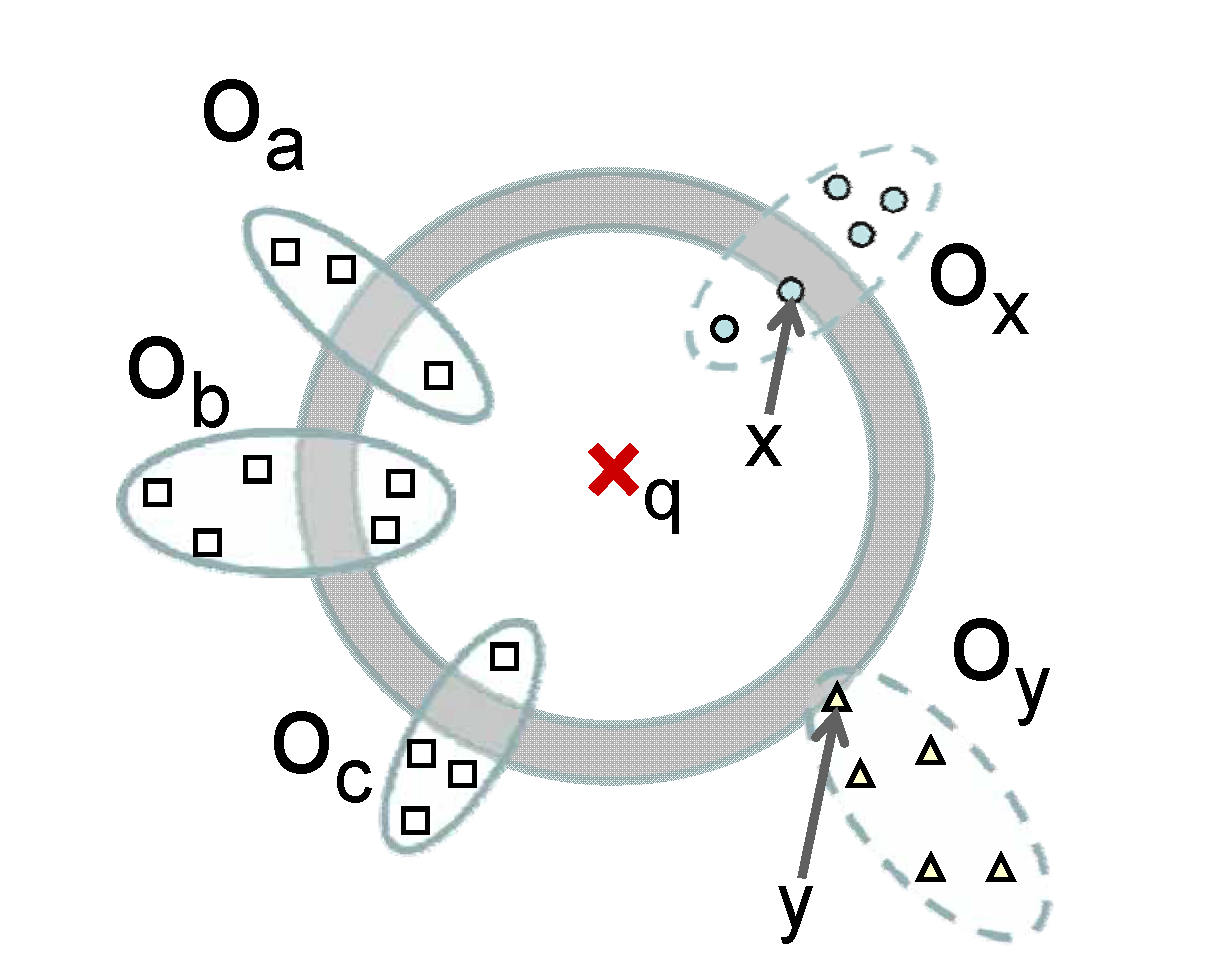}
        }
    \subfigure[Case 3: Instance $(y,p_y)$ is not the first
    returned instance of object $o_Y$.\label{fig:threeCases_case3}]{
        \includegraphics[width=0.29\columnwidth]{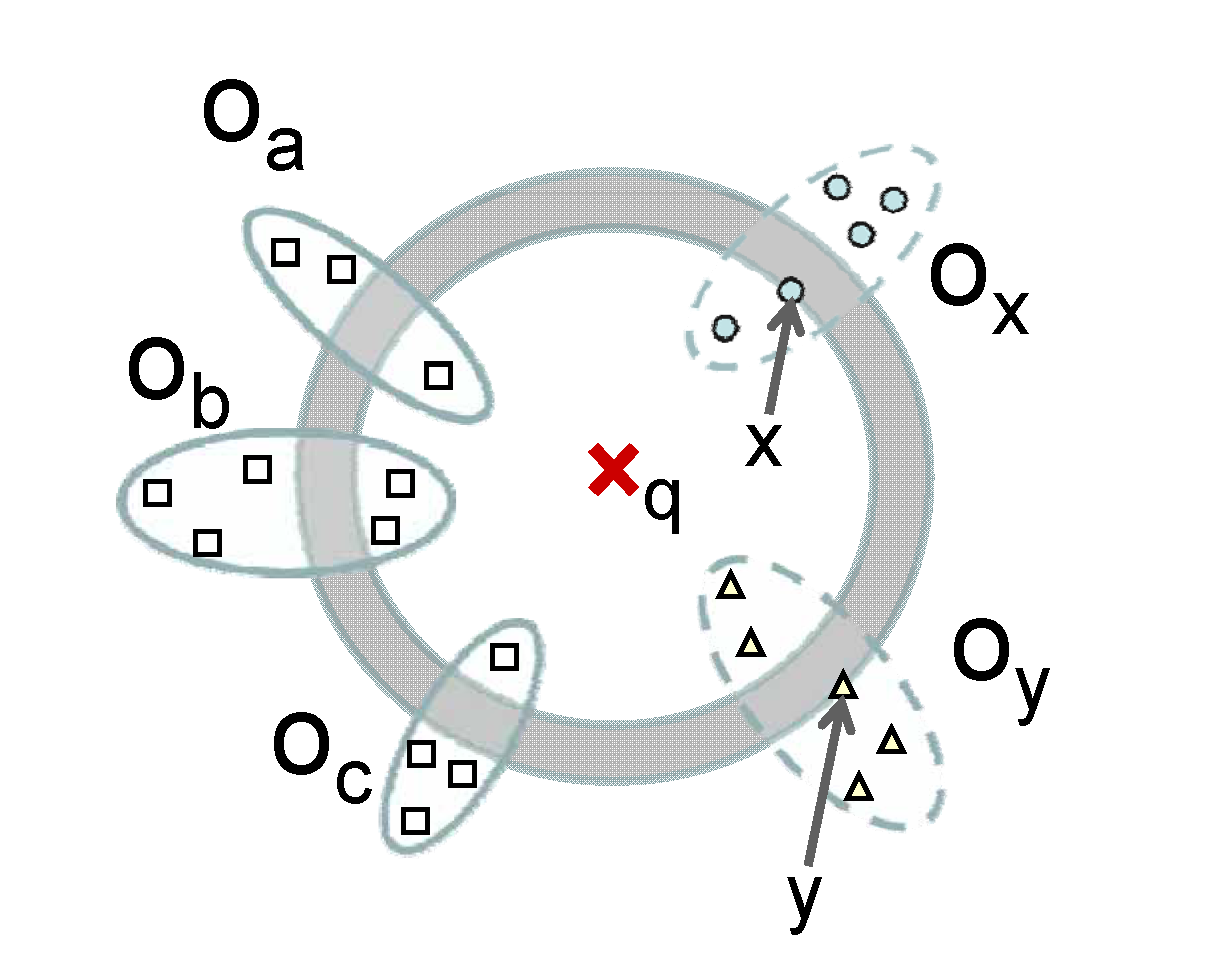}
        }
    \caption{Cases to be considered when updating the probabilities, assuming
    $x$ was the last processed instance and $y$ is the current one.}
    \label{fig:threeCases}
\end{figure*}

Now, we show how the probabilities $P_{i,\SSSY,y}$ for
$i\in\{0,\ldots,|\SSSY|\}$ can be computed in constant time
considering the above cases which are illustrated in Figure
\ref{fig:threeCases}.

In the first case (cf. Figure \ref{fig:threeCases_case1}), the
probabilities $P_x(o)$ and $P_y(o)$ of all  objects in $o \in
\SSSX$ are equal, because the instances of objects in $\SSSX$ that
appear within the distance range of $q$ of $y$ and within the
distance range of $x$ are identical. Since the probabilities
$P_{i,\SSSY,y}$ and $P_{i,\SSSX,x}$ only depend on the $P_x(o)$
for all objects $o\in\SSSX$, it is obvious that
$P_{i,\SSSY,y}=P_{i,\SSSX,x}$ for all $i$.

In the second case (cf. Figure \ref{fig:threeCases_case2}) we can
exploit the fact that $P_{i,\SSSX,x}$ does not depend on $o_Y$.
Thus, given the probabilities $P_{i,\SSSX,x}$, we can easily
compute the probability $P_{i,\SSSY,y}$ by incorporating the
object $o_X$ using the recursive Equation
\ref{equ:dynamicProgramming}:
$$
P_{i,\SSSY,y} =
$$
$$
P_{i-1,\SSSYX,y}\cdot P_y(o_X) + P_{i,\SSSYX,y}\cdot(1-P_y(o_X)).
$$
Since $\SSSYX = \SSSXY = \SSS -\{o_X,o_Y\}$ holds by definition
and no instance of any object in $\SSSXY$ appears within the
distance range of $q$ according to $y$ but not within the range
according to $x$, the following equation holds:
$$
P_{i,\SSSY,y} =
$$
$$
P_{i-1,\SSSXY,x}\cdot P_y(o_X) + P_{i,\SSSXY,x}\cdot(1-P_y(o_X)).
$$
Furthermore, $P_{i-1,\SSSXY,x} = P_{i-1,\SSSX,x}$, because $o_Y$
is not in the distance range according to $x$ and, thus,
$o_Y\notin\SSSX$. Now, the above equation can be reformulated:
\begin{equation}
\label{equ:caseTwo} P_{i,\SSSY,y} =
$$
$$
P_{i-1,\SSSX,x}\cdot P_y(o_X) + P_{i,\SSSX,x}\cdot(1-P_y(o_X)).
\end{equation}
All probabilities of the term on the right hand side in Equation
\ref{equ:caseTwo} are known and, thus, $P_{i,\SSSY,y}$ can be
computed in constant time assuming that the probabilities
$P_{i,\SSSX,x}$ computed in the previous step have been stored for
all $i\in\{0,\ldots,|\SSSX|\}$.

The third case (cf. Figure \ref{fig:threeCases_case3}) is the
general case which is not as straightforward as the previous two
cases and requires special techniques. Again, we assume that the
probabilities $P_{i,\SSSX,x}$ computed in the previous step for
all $i\in\{0,\ldots,|\SSSX|\}$ are known. Similar to Case 2, the
probability $P_{i,\SSSY,y}$ is equal to:
$$
P_{i,\SSSY,y} =
$$
\begin{equation}
\label{equ:P(i,S^oY,y)} P_{i-1,\SSSXY,x}\cdot P_y(o_X) +
P_{i,\SSSXY,x}\cdot(1-P_y(o_X)).
\end{equation}


Since the probability $P_y(o_X)$ is assumed to be known, now we
are left with the computation of $P_{i,\SSSXY,x}$ for all
$i\in\{0,\ldots,|\SSSXY|\}$ by again exploiting Equation
\ref{equ:dynamicProgramming}:
$$
P_{i,\SSSX,x}=
$$
$$
P_{i-1,\SSSXY,x}\cdot P_x(o_Y)+P_{i,\SSSXY,x}\cdot(1-P_x(o_Y))
$$
which can be resolved to
$$
P_{i,\SSSXY,x}=
$$
\begin{equation}
\label{equ:P(i,S^oX div
oY,x)}\frac{P_{i,\SSSX,x}-P_{i-1,\SSSXY,x}\cdot
P_x(o_Y)}{1-P_x(o_Y)}.
\end{equation}
With $i=0$ we have
$$
P_{0,\SSSXY,x}=\frac{P_{0,\SSSX,x}-P_{-1,\SSSXY,x}\cdot
P_x(o_Y)}{1-P_x(o_Y)}=
$$
$$
\frac{P_{0,\SSSX,x}}{1-P_x(o_Y)},
$$
because the probability $P_{-1,\SSSXY,x}=0$ by definition (cf.
Equation \ref{equ:dynamicProgramming}). The case $i=0$ can be
solved assuming that $P_{0,\SSSX,x}$ is known from the previous
iteration step.

With the assumption that all probabilities $P_{i,\SSSX,x}$ for all
$i\in\{1,\ldots,|\SSSX|\}$ and $P_x(o_Y)$ are available from the
previous iteration step, we can use Equation \ref{equ:P(i,S^oX div
oY,x)} to recursively compute $P_{i,\SSSXY,x}$ ($1\leq
i\leq|\SSSXY|$) using the previously computed $P_{i-1,\SSSXY,x}$.
Based on this recursive computation we obtain all probabilities
$P_{i,\SSSXY,x}$ ($0\leq i\leq|\SSSXY|$) which can used to compute
the probabilities $P_{i,\SSSY,y}$ for all $0\leq i\leq|\SSSXY|$
according to Equation \ref{equ:P(i,S^oY,y)}.

\subsection{Runtime Analysis}
\label{subsec:RuntimeProbabilisticRanking}

Building on this case-based analysis for the cost of computing
$P_{i,\SSSo,x}$ for the currently accessed instance $x$ of an
object $o$, we now prove that we can compute the rank
probabilities of all objects at cost $O(nk)$, where $n$ is the
number of object instances accessed. The following lemma suggests
that the incremental cost per object instance access is $O(k)$.

\begin{lemma}
\label{lem:constantTimeSpace} Let $(x,p_x)\in o_X$ and $(y,p_y)\in
o_Y$ be two object instances consecutively returned from the
distance browsing. W.l.o.g., let us assume that the instance
$(x,p_x)$ was returned in the last iteration in which we computed
the probabilities $P_{i,\SSSX,x}$ for all $0\leq i\leq|\SSSX|$.
The next iteration, in which we fetch $(y,p_y)$ the probabilities
$P_{i,\SSSY,y}$ for all $0\leq i\leq \min\{k,|\SSSY|\}$, can be
computed in $O(k)$ time and space.
\end{lemma}

\begin{proof}
In Case 1, the probabilities $P_{i,\SSSX,x}$ and $P_{i,\SSSY,y}$
are equal for all $0\leq i\leq min\{k,|\SSSY)|\}$. No computation
is required ($O(1)$ time) and the result can be stored using at
most $O(k)$ space.

In Case 2, the probabilities $P_{i,\SSSY,y}$ for all $0\leq i\leq
min\{k,|\SSSY)|\}$ can be computed according to Equation
\ref{equ:caseTwo} taking $O(k)$ time. This assumes that the
$P_{i,\SSSX,x}$ have to be stored for all $0\leq
i\leq\min\{k,|\SSSY|\}$, requiring at most $O(k)$ space.

In Case 3, we first have to compute and store the probabilities
$P_{i,\SSSXY,x}$ for all $0\leq i\leq\min\{k,|\SSSXY)|\}$ using
the recursive function in Equation \ref{equ:P(i,S^oX div oY,x)}.
This can be done in $O(\min\{k,|\SSSXY)|\})$ time and space. Next,
the computed probabilities can be used to compute $P_{i,\SSSY,y}$
for all $0\leq i\leq\min\{k,|\SSSY)|\}$ according to Equation
\ref{equ:P(i,S^oY,y)} which again takes at most
$O(\min\{k,|\SSSXY)|\})$ time and space.
\end{proof}

After giving the runtime evaluation of the processing of one
single object instance, we are now able to extend the cost model
for the whole query process. According to Lemma
\ref{lem:constantTimeSpace}, we can assume that each object
instance can be processed in constant time if we assume that $k$
is constant. If we assume that the total number of object
instances in our database is linear to the number of database
objects we would get a runtime complexity which is linear in the
number of database objects, more exactly particular O($kn$) where
$n$ is the size of the database and $k$ the specified depth of the
ranking. Up to now, our model assumes that the preprocessing step
and the postprocessing step of our framework requires at most
linear runtime. Since the postprocessing step only includes an
aggregation of the results generated in Step 2 the linear runtime
complexity of Step 3 is guaranteed. Now, we want to examine the
runtime of the object instance ranking in Step 1. Similar to the
assumptions that hold for our competitors
\cite{uncertain-sic-07,uncertain-ylks-08,bkr-08} we can also
assume that the object instances are already sorted, which would
involve linear runtime cost also for Step 1. However, for the
general case where we have to initialize a distance browsing
first, the runtime complexity of Step 1 would increase to
O($n$log$n$). As a consequence, the total runtime cost of our
approach (including distance browsing) sums up to
O($n$log$n$+$kn$). An overview of the computation cost is given in
Table \ref{table:complexity}.

Regarding the space complexity of our approach, we have to store,
for each object in the database, a vector of length $k$ for the
probabilistic ranking of size $O(kn)$. In addition, we have to
store the $AOL$ of at most size $O(n)$, yielding a total space
complexity of $O(kn+n)=O(kn)$. Note that \cite{uncertain-ylks-08}
computes a different ranking (cf. Section
\ref{sec:ProbabilisticRankingUncertainObjects} for details) with a
space complexity of $O(n)$. To compute a probabilstic ranking
according to our definition, \cite{uncertain-ylks-08} requires
$O(kn)$ space as well.

\begin{table}\centering
\resizebox {8.4cm }{!}{ %
\begin{tabular}{|c|c|c|}
  \hline
  runtime table & no precomputed $D$ & precomputed $D$ \\
  \hline
  ours & O($n$log$n$+$kn$) & O($kn$) \\
  \cite{uncertain-ylks-08} & O($kn^2$) & O($kn^2$)  \\
  \cite{bkr-08} & exponential & exponential \\
  \cite{uncertain-sic-07} & exponential & exponential \\
  \hline
\end{tabular}}
\caption{Runtime complexity comparison of the best-known
approaches to our own aproach. } \label{table:complexity}

\end{table}

\section{Probabilistic Ranking Algorithm}
\label{sec:Algorithm}

\begin{figure}[t]
\centering \fbox{
  \begin{minipage}{0.98\columnwidth}
  \sf\small
  Probabilistic Ranking(\emph{D},\emph{q})\\[1.2ex]
  1\phantom{0}\quad\hspace*{0em}\emph{AOL} = $\emptyset$ \\
  2\phantom{0}\quad\hspace*{0em}result = Matrix of zeros // size:
  $\mid$instances$\mid$*k\hspace{18pt}\\
  3\phantom{0}\quad\hspace*{0em}\emph{p-rank}\_x = [0,\ldots,0] // Length \emph{k} \\
  4\phantom{0}\quad\hspace*{0em}\emph{p-rank}\_y = [0,\ldots,0] // Length \emph{k} \\
  5\phantom{0}\quad\hspace*{0em} \\
  6\phantom{0}\quad\hspace*{0em}\emph{y} = \emph{D.next} \\
  7\phantom{0}\quad\hspace*{0em}\emph{updateAOL(y)}\\
  8\phantom{0}\quad\hspace*{0em}\emph{p-rank}\_x[0]=1\\
  9\phantom{0}\quad\hspace*{0em}Add \emph{p-rank}\_x to the first line of
  result.\\
  10\quad\hspace*{0em}\textbf{FOR} (\emph{D} is not empty \textbf{AND} $\exists
  p\in$ \emph{p-rank}\_x: $p>0$)\\
  12\quad\hspace*{0.3cm}\emph{x} = \emph{y} \\
  13\quad\hspace*{0.3cm}\emph{y} = \emph{D.next} \\
  14\quad\hspace*{0.3cm}\emph{updateAOL(y)} \\
  15\phantom{0}\quad\hspace*{0em} \\
  16\quad\hspace*{0.3cm}\textbf{CASE 1: (c.f. Figure
  \ref{fig:threeCases_case1})}\\
  17\quad\hspace*{0.3cm}\textbf{IF} (\emph{$o_y$} =  \emph{$o_x$})\\
  18\quad\hspace*{0.6cm}\emph{p-rank}\_y =  \emph{p-rank}\_x \\
  19\quad\hspace*{0.3cm}\textbf{END-IF}\\
  20\phantom{0}\quad\hspace*{0em} \\
  21\quad\hspace*{0.3cm}\textbf{CASE 2: (c.f. Figure
  \ref{fig:threeCases_case2})}\\
  22\quad\hspace*{0.3cm}\textbf{ELS-IF} ($o_y
  \not \in AOL$)\\
  23\quad\hspace*{0.6cm}$P(o_x)$=AOL.getProb($o_x$)\\
  24\quad\hspace*{0.6cm}\emph{p-rank}\_y = \emph{dynamicRound}(\emph{p-rank}\_x,$P(o_x)$)])\\
  25\quad\hspace*{0.3cm}\textbf{END-IF}\\
  26\phantom{0}\quad\hspace*{0em} \\
  27\quad\hspace*{0.3cm}\textbf{CASE 3: (c.f. Figure
  \ref{fig:threeCases_case3})}\\
  28\quad\hspace*{0.3cm}\textbf{ELSE} \phantom{blablabla} // (\emph{$o_y$} !=
  \emph{$o_x$})\\
  29\quad\hspace*{0.6cm}$P(o_x)$=AOL.getProb($o_x$)\\
  30\quad\hspace*{0.6cm}$P(o_y)$=AOL.getProb($o_y$)\\
  31\quad\hspace*{0.6cm}\emph{adjustedProbs} =
  \emph{adjustProbs(prevProbs,$P(o_y)$)} \\
  32\quad\hspace*{0.6cm}\emph{p-rank}\_y =
  \emph{dynamicRound(adjustedProbs,$P(o_x)$)}\\
  33\quad\hspace*{0.3cm}\textbf{END-IF}\\
  34\phantom{0}\quad\hspace*{0em} \\
  35\quad\hspace*{0.3cm}Add
  \emph{p-rank}\_y to the next line of result.\\
  36\quad\hspace*{0.3cm}\emph{p-rank}\_x = \emph{p-rank}\_y \\
  37\quad\hspace*{0em}\textbf{END-FOR}\\
  38\quad\hspace*{0em}return \emph{result}\\
  39\quad\textbf{END} Probabilistic Ranking.\\
  \end{minipage}
} \caption{Pseudocode of our ranking algorithm.}
\label{fig:alg-rknn}
\end{figure}

\begin{figure}[t]
\centering \fbox{
  \begin{minipage}{0.98\columnwidth}
  \sf\small
  dynamicRound(\emph{oldRanking},\emph{$P(o_X)$})\\[1.2ex]
  1\quad\phantom{0}\hspace*{0em}\emph{newRanking} = [0,\ldots,0] // Length
  \emph{k} \\
  2\quad\phantom{0}\hspace*{0em}\emph{newRanking}[0] = \\
  3\quad\phantom{0}\hspace*{0.6cm}\emph{oldRanking}[0]*(1-\emph{$P(o_X)$})\\
  4\quad\phantom{0}\hspace*{0em}\textbf{FOR} \emph{i} =
  \emph{1},\ldots,\emph{k-1} \\
  5\quad\phantom{0}\hspace*{0.3cm}\emph{newRanking}[i] =\\
  6\quad\phantom{0}\hspace*{0.9cm}\emph{oldRanking}[i-1]*\emph{$P(o_X)$}\\
  7\quad\phantom{0}\hspace*{0.9cm}+\emph{oldRanking[i]}*(1-\emph{$P(o_X)$})
  \\
  8\quad\phantom{0}\hspace*{0em}\textbf{END-FOR}\\
  9\quad\phantom{0}\hspace*{0em} return \emph{newRanking}\\
  10\quad\textbf{END} dynamicRound.\\
  \end{minipage}
} \caption{Pseudocode of a single dynamic Iteration.}
\label{fig:dynamicRound}
\end{figure}

\begin{figure}[t]
\centering \fbox{
  \begin{minipage}{0.98\columnwidth}
  \sf\small
  adjustProbs(\emph{oldRanking},$P(o_X)$)\\
  1\quad\phantom{0}\hspace*{0em}\emph{adjustedRanking} = [0,\ldots,0] // Length
  \emph{k} \\
  2\quad\phantom{0}\hspace*{0em}\emph{adjustedProbs}[0] = \\
  3\quad\phantom{0}\emph{oldRanking}[0] / \emph{$P(o_X)$}\\
  4\quad\phantom{0}\hspace*{0em}\textbf{FOR} \emph{i} =
  \emph{1},\ldots,\emph{k-1} \\
  5\quad\phantom{0}\hspace*{0.3cm}\emph{adjustedProbs}[\emph{i}] =$\frac{oldRanking[i]-oldRanking[i-1]*P(o_X)}
  {(1-P(o_X))}$\\
  6\quad\phantom{0}\hspace*{0em}\textbf{END-FOR}\\
  7\quad\phantom{0}\hspace*{0em}return \emph{adjustedProbs}\\
  8\quad\phantom{0}\textbf{END} adjustProbs.\\
  \end{minipage}
} \caption{Pseudocode of the algorithm that computes the
$P_{i,\SSSXY,x}$ from the $P_{i,\SSSX,x}$ for all
$i\in\{0,\ldots,k-1\}$.} \label{fig:adjustProbs}
\end{figure}

The pseudocode of the algorithm for the probabilistic ranking is
illustrated in Figure \ref{fig:alg-rknn}, providing the
implementation details of the previously discussed steps. Our
algorithm requires a query object $q$ and a distance browsing
operator $D$ (cf. \cite{hs-95}), that allows us to iteratively
access the object instances sorted in ascending order of their
similarity distance to a query object.

First, we initialize the \emph{activeObjectList (AOL)} , a data
structure that contains one tuple $(o,p_o)$ for each object $o$
that
\begin{itemize}
  \item has previously been found in $D$, i.e. at least one instance of $o$ has been processed and
  \item has not yet been completely processed, i.e. at least one
  instance of $o$ has yet to be found,
\end{itemize}
associated with the sum $p_o$ of probabilities of all its
instances that have been found. The $AOL$ offers two
functionalities:
\begin{itemize}
  \item updateAOL(instance $i$): Adds to the probability of $i$ to $p_o$,
  where $o$ is the object that $i$ belongs to.
  \item getProb(object o): Returns $p_o$.
\end{itemize}

Note that it is mandatory that the position of a tuple $(o,p_o)$
can be found in constant time, in order to sustain the constant
time complexity of an iteration. This can be
\begin{itemize}
  \item approached by means of hashing or
  \item reached by giving each object $o$ the information about the location of
  its corresponding tuple $p_o$ at an additional space cost of $O(n)$.
\end{itemize}

We also keep the \emph{result}, a matrix that contains, for each
object instance $x$ that has been found and each ranking position
$i$, the probability $\P_i(x)$ that $x$ is located at ranking
position $i$. Note that this result is instance-based. In order to
get an object-based rank probability, we can aggregate intances
belonging to the same object, using Equation
\ref{equ:sumOfInstanceProbs}. Additionally, we initialize two
arrays \emph{p-rank}\_x and \emph{p-rank}\_y, each of length $k$,
which contain, at any iteration of the algorithm, the
probabilities $P_{i,\SSSX,x}$ and $P_{i,\SSSY,y}$ respectively,
for all $0\leq i\leq k$. $x \in o_X$ is the instance found in the
previous iteration and $y \in o_Y$ is the instance found in the
current iteration (see Figure \ref{fig:threeCases}).

In line 6, the algorithm starts by fetching the first object
instance, which is closest to the query $q$ in the database. A
tuple containing the corresponding object as well as the
probability of this instance is added to the {\it AOL}.

Then, the first position of \emph{p-rank}\_x is set to $1$ while
all other $k-1$ positions remain at $0$, because
$$
P_{1,\SSSCI,y}=\\P_{1,\emptyset,y}=1
$$
and
$$
P_{i,\SSSCI,y}=\\P_{i,\emptyset,y}=0
$$
for $i>1$ by definition (see Equation
\ref{equ:dynamicProgramming}). This simply reflects the fact that
the first instance is always on rank $1$. Note that
\emph{p-rank}\_y is implicitely assigned to \emph{p-rank}\_x here.

Then, the first iteration of the main algorithm begins by fetching
the next object instance from $D$. Now, we have do distinguish the
three cases explained in Section \ref{sec:probsimrankingmethods}.

In the first case (line 16), both the previous and the current
instance refer to the same object. As explained in Section
\ref{sec:probsimrankingmethods}, we have nothing to do in this
case, since $P_{i,\SSSX,x}$=$P_{i,\SSSY,y}$ for all $0\leq i\leq
k-1$.

In the second case (line 21), the current instance refers to an
object that has not been seen yet. As explained in Section
\ref{sec:probsimrankingmethods}, we only have to apply an
additional iteration of the DP algorithm (cf. Equation
\ref{equ:dynamicProgramming}). This \emph{dynamicRound} algorithm
is shown in Figure \ref{fig:dynamicRound} and is used here to
incorporate the probability that $o_x$ is closer to $y$ into
\emph{p-rank}\_y in a single iteration of the dynamic algorithm.

In the third case (line 27), the current instance relates to an
object that has already been seen. Thus the probabilities
$P_{i,\SSSX,x}$ depend on $o_Y$. As explained in Section
\ref{sec:probsimrankingmethods}, we first have to filter out the
influence of $o_Y$ on $P_{i,\SSSX,x}$ and compute
$P_{i,\SSSXY,x}$. This is performed by the \emph{adjustProbs}
algorithm in Figure \ref{fig:adjustProbs} utilizing the technique
explained in Section \ref{sec:probsimrankingmethods}. Using the
$P_{i,\SSSXY,x}$, the algorithm then computes the $P_{i,\SSSY,y}$
using a single iteration of the dynamic algorithm like in case
two.

At line 35, the computed ranking for instance $y$ is added to the
result. If the application (i.e. the ranking method) requires
objects to be ranked instead of instances, then \emph{p-rank}\_y
is used to incrementally update the probabilities of $o_y$ for
each rank.

The algorithm continues fetching object instances from the
distance browsing operator $D$ and repeats this case analysis
until either no more samples are left in $D$ or until an object
instance is found, that has no influence on the $k$ first ranking
positions.

\section{Probabilistic Ranking Approaches}
\label{sec:ProbabilisticRankingUncertainObjects}

The method proposed in Section \ref{sec:probsimrankingmethods}
efficiently computes for each uncertain object instance $o_j$ and
each ranking position $i$ ($0\leq i\leq k-1$) the probability that
$o_j$ has the $i^{th}$ rank. However, most applications require an
unique object ranking, i.e. each object (or object instance) is
uniquely assigned to exactly one rank. Various top-$k$ query
approaches have been proposed generating deterministic rankings
from probabilistic data which we call probabilistic ranking
queries. The question at issue is how our framework can be
exploited in order to significantly accelerate probabilistic
ranking queries. In the remainder, we show that our framework is
able to support and significantly boost the performance of the
state-of-the-art probabilistic ranking queries. Specifically, we
demonstrate this by applying state-of-the-art ranking approaches
including, U-$k$Ranks, PT-$k$ and \emph{Global top-k}.

Note, that the following ranking approaches are based on the
x-relation model \cite{bshw-06,abshnsw-06}. As mentioned before,
the x-relation model conceptionally corresponds to our uncertainty
model, where the object instances correspond to the tuples and the
uncertain vector objects correspond to the x-tuples. In the
following, we use the terms \emph{object instance} and
\emph{object}.

\subsection{Expected Score and Expected Ranks}
The \emph{Expected Score} and \emph{Expected Ranks} \cite{cly-09}
compute for each object instance its expected score (rank) and
rank the instances by this expected score (rank). \emph{Expected
Ranks} runs in $O(n\cdot log(n))$-time, thus outperforming exact
approaches that do not use any estimation. The main drawback of
this approach is that by using the expected value estimator,
information is lost about the distribution of the objects. In the
following, we will show how our framework can be used to
accelerate the remaining state-of-the-art approaches, including
U-$k$Ranks, PT-$k$ and \emph{Global top-k}, to $O(n\cdot log n +
kn)$ runtime.

\begin{figure}[t]
        \centering
    \hspace*{\fill}

        \includegraphics[width=0.25\columnwidth]{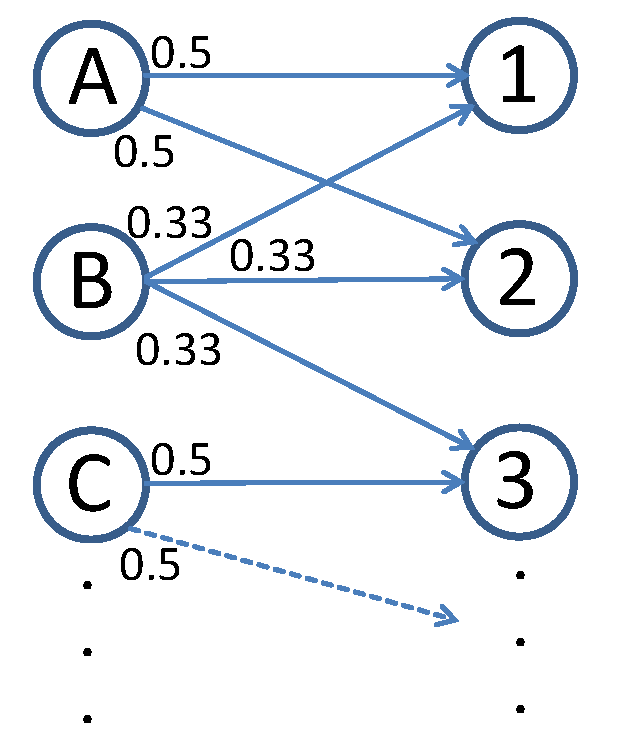}
    \caption{Small example extract of a probabilistic ranking as produced by our
    framework.}
    \label{fig:probrankingexample}
\end{figure}

\subsection{U-$k$Ranks}
The U-$k$Ranks \cite{uncertain-sic-07} approach reports the most
likely object instance at each rank $i$, i.e. the instance that is
most likely to be ranked $i$th over all possible worlds. This is
essentially the same definition as proposed in $PRank$ in
\cite{cl-08} in the context of distributions over spatial data.
The approach proposed in \cite{uncertain-sic-07} has exponential
runtime. The runtime has been reduced to $O(n^2k)$ time in
\cite{tkde-ylks-08}. Using our framework, the problem of
U-$k$Ranks can be solved in $O(n\cdot log(n) + nk)$ time using the
same space
complexity as follows:\\
Use the framework to create the probabilistic ranking in $O(n\cdot
log(n) + nk)$ as explained in the previous section. Then, for each
rank $i$, find the object instance $argmax_{j}(p\_rank_q(o_j,i))$
that has the highest probability of appearing at rank $i$ in
$O(nk)$. This is performed by (cf. Figure
\ref{fig:probrankingexample}) finding for each rank $i$ the object
instance which has the highest probability to be assigned to rank
$i$. Obviously, a problem of this problem definition is that a
single object instance $o_j$ may appear at more than one ranking
position, or at no ranking position at all. For example in
\ref{fig:probrankingexample}, object instance $A$ is ranked on
both ranks $1$ and $2$, while object instance $B$ is ranked
nowhere. The total runtime for U-$k$Ranks has thus been reduced
from $O(n^2)$ to $O(nlog(n)+kn)$, that is $O(n*log(n))$ if $k$ is
assumed to be constant.

\subsection{PT-$k$}
The \emph{probabilistic threshold top-k} query (PT-$k$)
\cite{hpzl-08} problem fixes the problem of the previous
definition by aggregating the probabilities of an object instance
$o_j$ appearing at rank $k$ or better. Given a user-specified
probability threshold $p$, PT-$k$ returns all instances, that have
a probability of at least $p$ of being at rank $k$ or better. Note
that in this definition, the number of results is not limited to
$k$ and depends on the threshold parameter $p$. The model of
PT-$k$ consists of a set of instances and a set of generation
rules that define mutually exclusiveness of instances. Each object
instance occurs in one and only one generation rule. This model
conceptionally corresponds to the x-relation model (with disjoint
x-tupels). PT-$k$ computes all result instances in $O(nk)$ time
while also assuming that the instances are already pre-sorted, thus having a total runtime of $O(nlog(n)+kn)$. \\
The framework can be used to solve the PT-$k$ problem in the
following way:\\ We create the probabilistic ranking in $O(nk)$ as
explained in the previous section. For each object instance $o_j$,
we compute the probability that $o_j$ appears at position $k$ or
better (in $O(nk)$). Formally, we return all instances $o_j\in\DB$
for which:
$$
\{o_j\in\DB| \sum_{i=1}^{k}p\_rank_q(o_j,i)>p\}
$$
As seen in Figure \ref{fig:probrankingexample}, this probability
can simply be computed by aggregating all probabilities of an
object instance to be ranked at $k$ or better. For example, for
$k=2$ and $p=0.5$, we get $A$ and $B$ as results. Note that for
$p=0.1$, further object instances may be in the result, because
there must be further object instances (from object instances that
are left out here for simplicity) with a probability greater than
zero to rank $1$ and rank $2$, since the probability of their
respective edges does not sum up to $1.0$ yet.

Note that our framework is only able to match, not to beat the
runtime of PT-$k$. However, using our approach, we can
additionally return the ranking order, instead of just the top-$k$
set.

\subsection{Global top-$k$}
\emph{Global top-k} \cite{zc-08} is very similar to PT-$k$ and
ranks the object instances by their top-$k$ probability, and then
takes the top-$k$ of these. This approach has a runtime of
$O(n^2k)$. The advantage here is that, unlike in PT-$k$, the
number of results is fixed, and there is no user-specified
threshold parameter. Here we can exploit the ranking order
information that we acquired in the PT-$k$ using our framework to
solve
\emph{Global top-k} in $O(n\cdot log(n)+kn)$ time: \\
We use the framework to create the probabilistic ranking in
$O(n\cdot log(n)+kn)$ as explained in the previous section. For
each object instance $o_j$, we compute the probability that $o_j$
appears at position $k$ or better (in $O(nk)$) like in PT-$k$.
Then, we find the $k$ object instances with the highest
probability in $O(k\cdot log(k))$.

\section{Experimental Evaluation}
\label{sec:Experiments}

We have performed extensive experiments to evaluate the
performance of our proposed probabilistic ranking approach
proposed in Section \ref{sec:probsimrankingmethods} w.r.t. the
database size ($|\DB|$) measured in the number of uncertain vector
objects, ranking depth ($k$) and degree of uncertainty (UD) as
defined below. In the following, the ranking framework is briefly
denoted by \PRaUD.

\begin{figure*}[ht]
\centering \subfigure[\Feifei]{
\includegraphics[width=0.29\columnwidth]{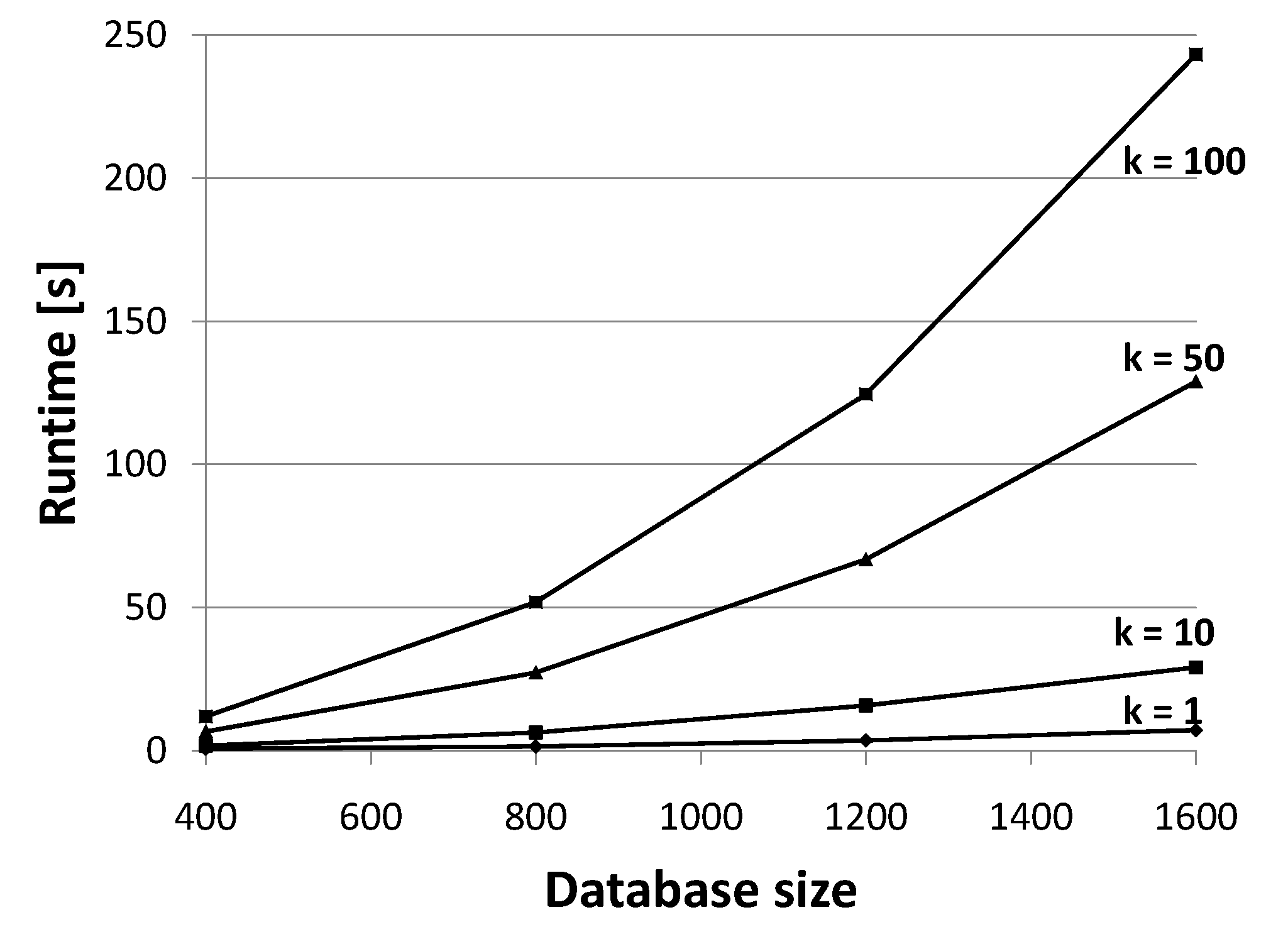}
\label{fig:eval_dbsize_env_feifei} }
\hspace{1em}\subfigure[\PRaUD]{
\includegraphics[width=0.29\columnwidth]{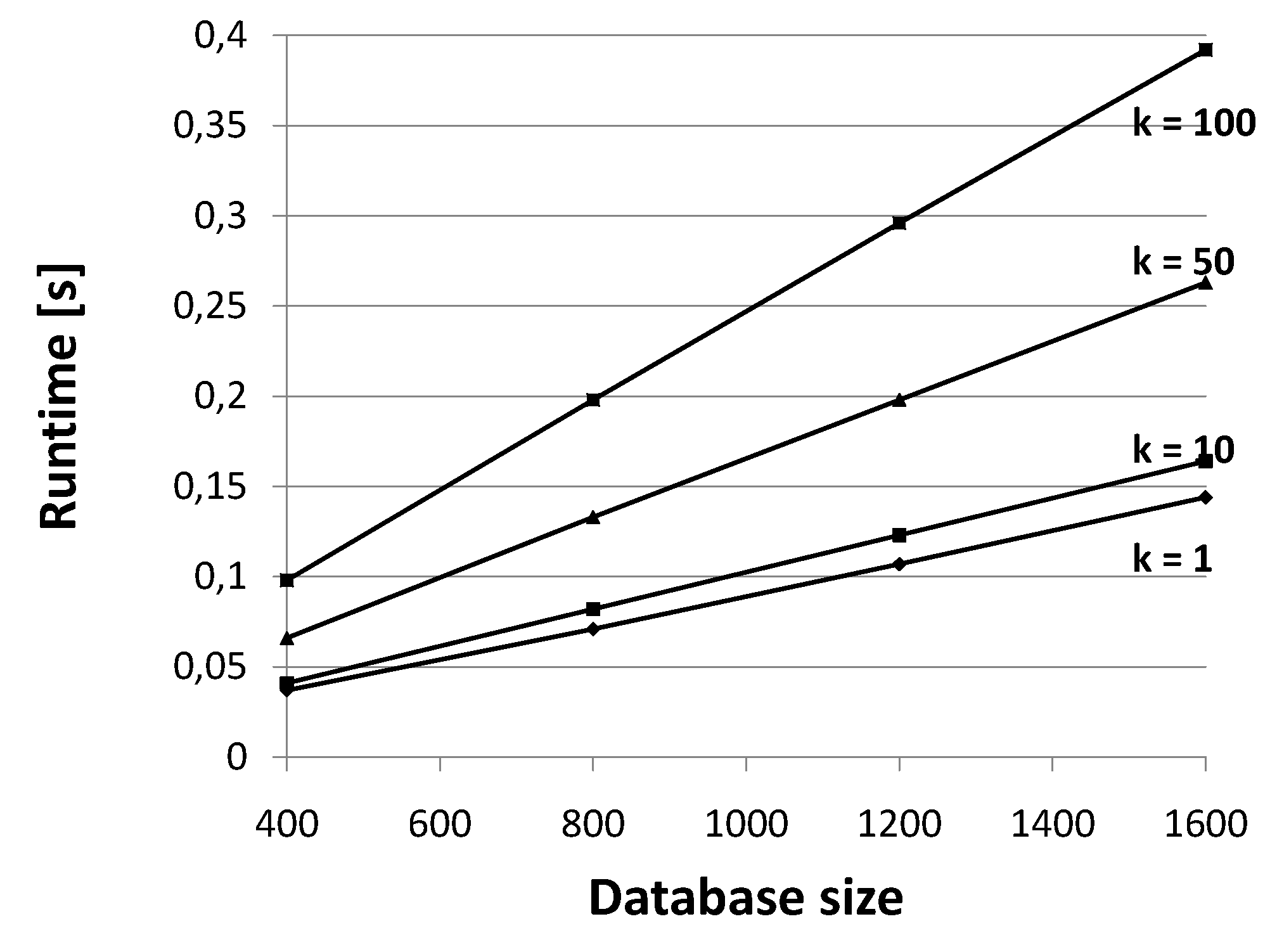}
\label{fig:eval_dbsize_env_praud} }
\hspace{1em}\subfigure[Speed-up gain w.r.t. $k$ on \emph{SCI}.]{
\includegraphics[width=0.29\columnwidth]{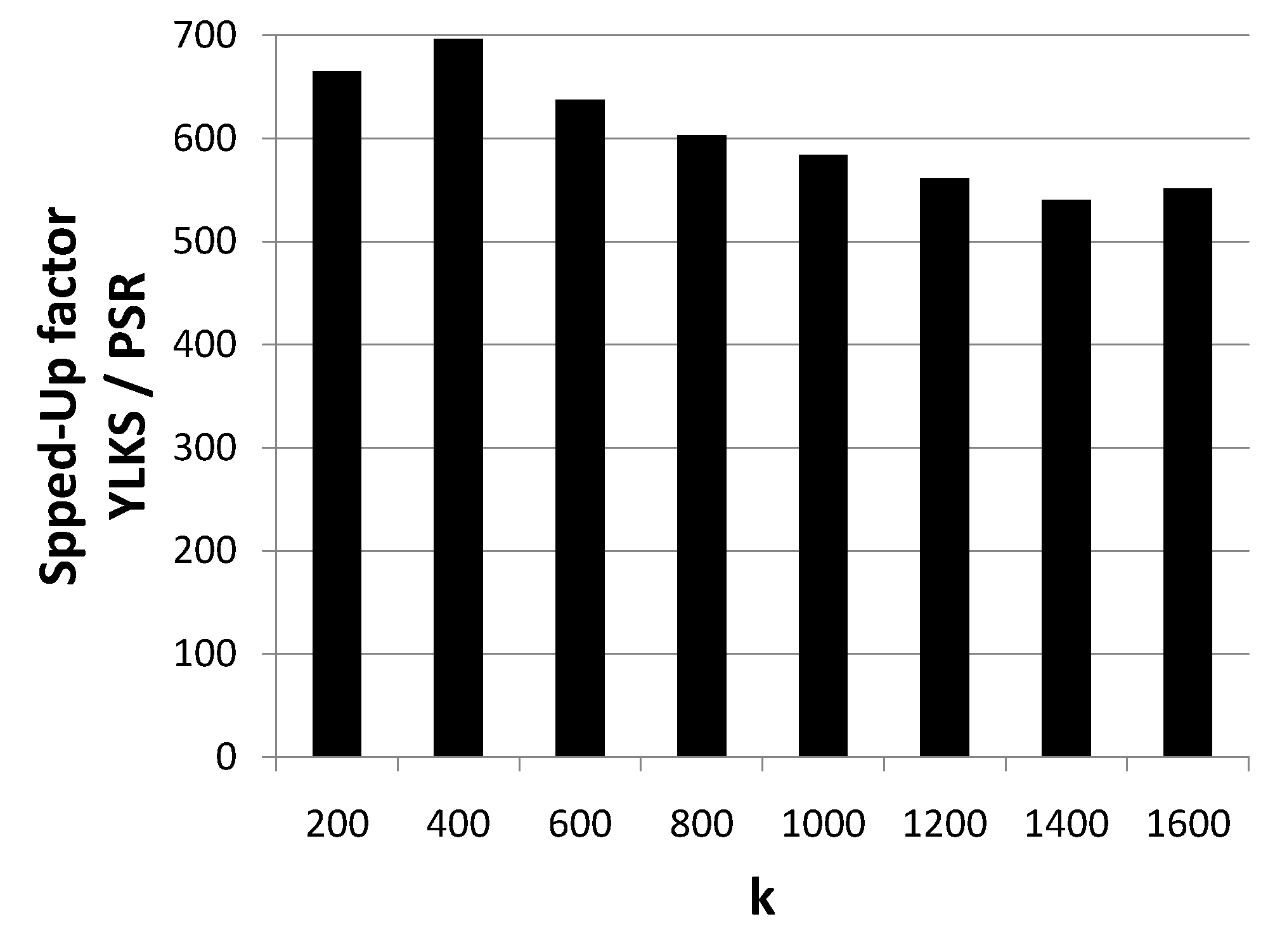}
\label{fig:eval_dbsize_env_speedup} }
\hspace{\fill}\caption{Scalability evaluated on \emph{SCI} having
different values for $k$.} \label{fig:eval_dbsize_env}
\end{figure*}

\subsection{Datasets and Experimental Setup}
\label{subsec:datasets}

The probabilistic ranking was applied to a scientific real-world
dataset \emph{SCI} and several artificial datasets \emph{ART\_X}
of varying size and degree of uncertainty. All datasets are based
on the discrete uncertainty model, i.e. each object is represented
by a collection of vector samples.

The \emph{SCI} dataset is a set of 1600 objects where each object
consists of 48 10-dimensional instances. Each instance corresponds
to a set of environmental sensor measurements of one single day
(one per 30 minutes) that consist of ten dimensions (attributes):
Temperature, humidity, speed and direction of wind w.r.t. degree
and sector, as well as concentrations of $CO$, $SO_2$, $NO$,
$NO_2$ and $O_3$. These attributes are normalized within the
interval [0,1] to give each attribute the same weight.

The \emph{ART\_1} dataset consists of up to 1,000,000 objects for
the scalability experiments. For the evaluation of the performance
w.r.t. the ranking depth and the degree of uncertainty we applied
a collection \emph{ART\_2} of datasets each composing 10,000
objects. Each object is represented by a set of 20 3-dimensional
instances. The \emph{ART\_2} datasets differs in the degree of
uncertainty ($UD$) the corresponding objects have. The degree of
uncertainty ($UD$) reflects the following distribution of object
instances: each uncertain vector object is assumed to be located
within an 3-dimensional hyper-rectangle. The object instances are
uniformly distributed within the corresponding rectangle. In the
following, we will refer to the side length of the rectangles as
\emph{degree of uncertainty} ($UD$). The rectangles are uniformly
distributed within a $10\times 10\times 10$ vector space.

The degree of uncertainty is interesting in our performance
evaluation since it is expected to have a significant influence on
the runtime. The reason is that a higher degree of uncertainty
obviously leads to an higher overlap between the objects which
influences the size of the active object list (AOL) (cf. Section
\ref{sec:Algorithm}) during the distance browsing. The higher the
object overlap the more objects are expected to be in the AOL at a
time. Since the size of the AOL influences the runtime of the rank
probability computation, a higher degree of uncertainty is
expected to lead to a higher runtime. This is experimentally
evaluated in Section \ref{subsec:uncertainty}.

\subsection{Scalability}
\label{subsec:runtime}

In this section, we give an overview of our experiments regarding
the scalability of \PRaUD. We compare our results to the dynamic
programming based rank probability computation used for the
U-kRanks method as proposed by Yi et al. in
\cite{uncertain-ylks-08}. This method, in the following denoted by
\Feifei, is the best approach currently known for solving the
(instance-based) rank probability problem (cf. Table
\ref{table:complexity}). For a fair comparison, we used the
\PRaUD\ framework to compute the same (instance-based) rank
probability problem as described in Section
\ref{sec:probsimrankingmethods}. Let us note that the cost
required to solve the object-based rank probability problem is
similar to that required to solve the instance-based rank
probability problem. This is because the former problem
additionally only requires to build the sum over all
instance-based rank probabilities which can be done on-the-fly
without additional cost. Furthermore, we can neglect the cost
required to build a final definite ranking (e.g. the rankings
proposed in Section
\ref{sec:ProbabilisticRankingUncertainObjects}) from the rank
probabilities, because they can be also computed on-the-fly by
simple aggregations of the corresponding (instance-based) rank
probabilities.

For the sorting of the distances of the instances to the query
point, we used a tuned quicksort adapted from \cite{bedu-93}. This
algorithm offers $O(n\cdot log(n))$ performance on many data sets
that cause other quicksorts to degrade to quadratic performance.

\begin{figure*}[ht]
\centering \subfigure[\Feifei]{
\includegraphics[width=0.29\columnwidth]{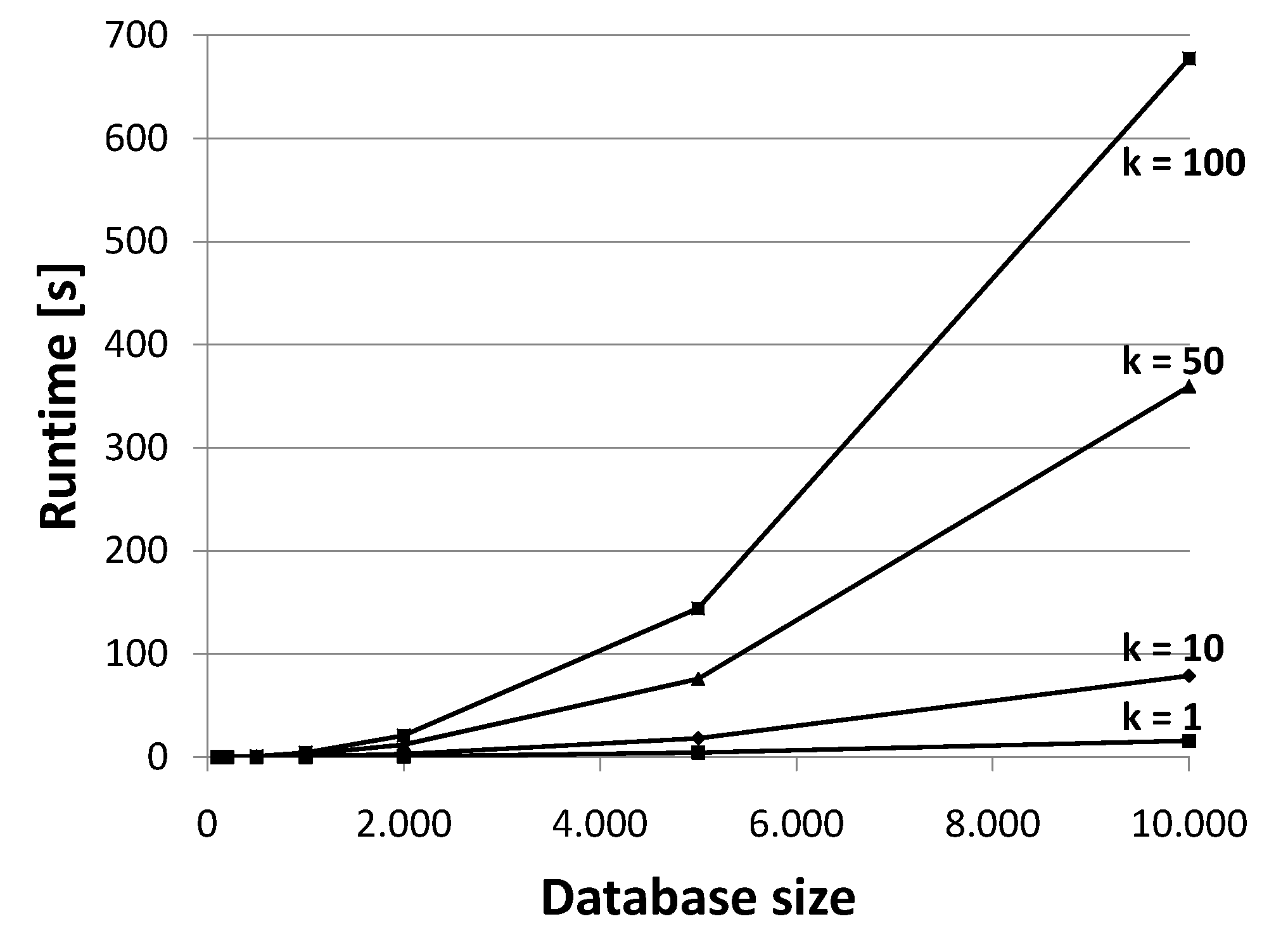}
\label{fig:eval_dbsize_art_feifei} }
\hspace{1em}\subfigure[\PRaUD]{
\includegraphics[width=0.29\columnwidth]{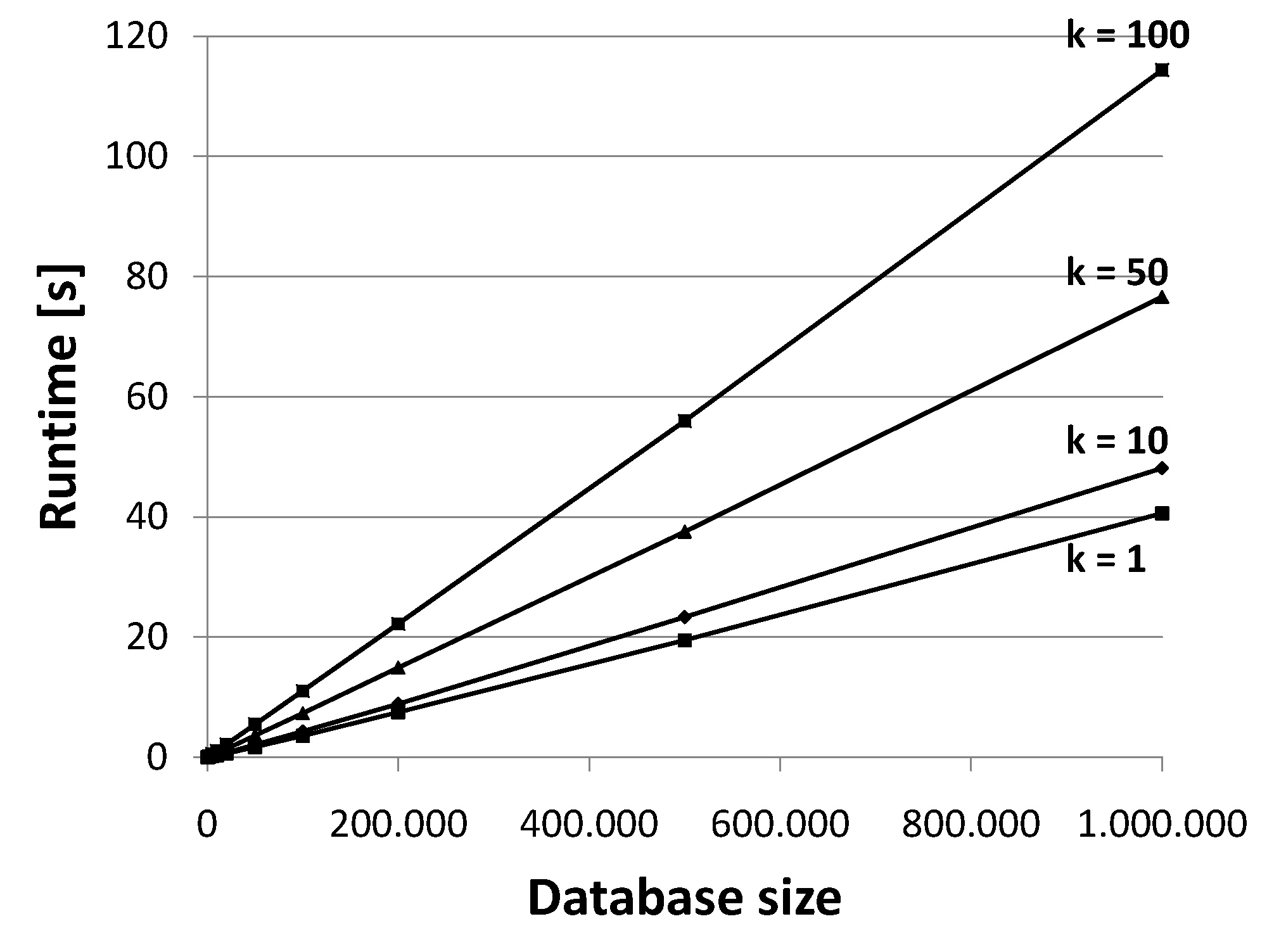}
\label{fig:eval_dbsize_art_praud} }
 \hspace{1em}\subfigure[Speed-up gain for an increasing $k$ grouped by an
 ascending number of objects in the database (\emph{ART\_1} dataset).]{
 \includegraphics[width=0.29\columnwidth]{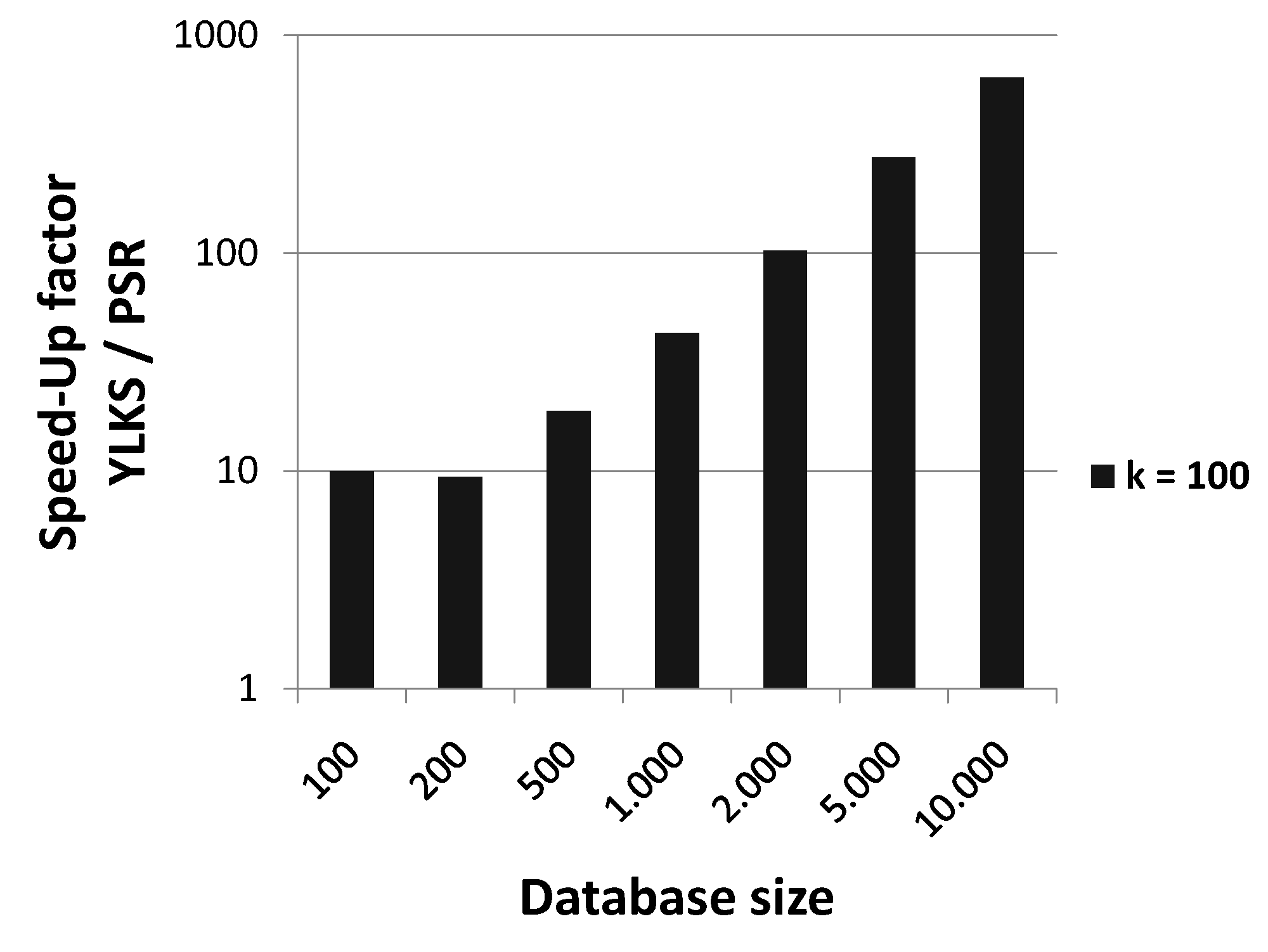}
 \label{fig:eval_dbsize_art_speedup}
 }
\hspace{\fill}\caption{Scalability evaluated on \emph{ART\_1}
w.r.t. $k$.} \label{fig:eval_dbsize_art}
\end{figure*}

The results of our first scalability tests on the real-data set
\emph{SCI} are depicted in Figure \ref{fig:eval_dbsize_env}. It
can be observed in Figure \ref{fig:eval_dbsize_env_praud} that the
runtime of the probabilistic ranking using the \PRaUD\ framework
increases linearly in the database size, whereas \Feifei\ has a
runtime quadratic in the database size in the same parameter
settings (cf. Figure \ref{fig:eval_dbsize_env_feifei}). We can
also see that this effect persists for different settings of $k$.
Note that the effect of the $O(n\cdot log(n))$ sorting of the
distances of the instances is insignificant on this relatively
small dataset. The direct speed-up of the rank probability
computation using \PRaUD\ in comparison to \Feifei\ is depicted in
Figure \ref{fig:eval_dbsize_env_speedup}. It shows for different
values of $k$, the speed-up factor, that is defined as the ratio
$\frac{runtime(\Feifei)}{runtime(\PRaUD)}$ describing the
performance gain of \PRaUD\ vs. \Feifei. It can be observed that,
for a constant number of objects in the database ($|DB|=1600$),
the ranking depth $k$ has no impact on the speed-up factor. This
can be explained by the observation that both approaches scale
linear in $k$.

Next, we evaluate the scalability of the database size based on
the \emph{ART\_1} dataset. The results of this experiment are
depicted in Figure \ref{fig:eval_dbsize_art}. Figure
\ref{fig:eval_dbsize_art_praud} shows that we are able to perform
ranking queries in a reasonable time of less than 120 seconds,
even for very large database containing 1,000,000 and more
objects, each having 20 instances (thus having a total of
20,000,000 instances (tupels)). Note that the time required to
sort the instances (less than 10 seconds for all 1,000,000
objects) is still insignificant compared to the total query cost.
In Figure \ref{fig:eval_dbsize_art_feifei}, it can be observed,
that due to the quadratic scaling of the \Feifei\ algorithm, it is
inapplicable for relatively small databases of size $5000$ or
more. The direct speed-up of the rank probability computation
using \PRaUD\ in comparison to \Feifei\ for varying database size
is depicted in Figure \ref{fig:eval_dbsize_art_speedup}. Here, we
can see that the speed-up of our approach in comparison to
\Feifei\ increases linear with the size of the database which is
consistent with our runtime analysis in Section
\ref{sec:probsimrankingmethods}.

\begin{figure}[ht]
\centering
\includegraphics[width=0.5\columnwidth]{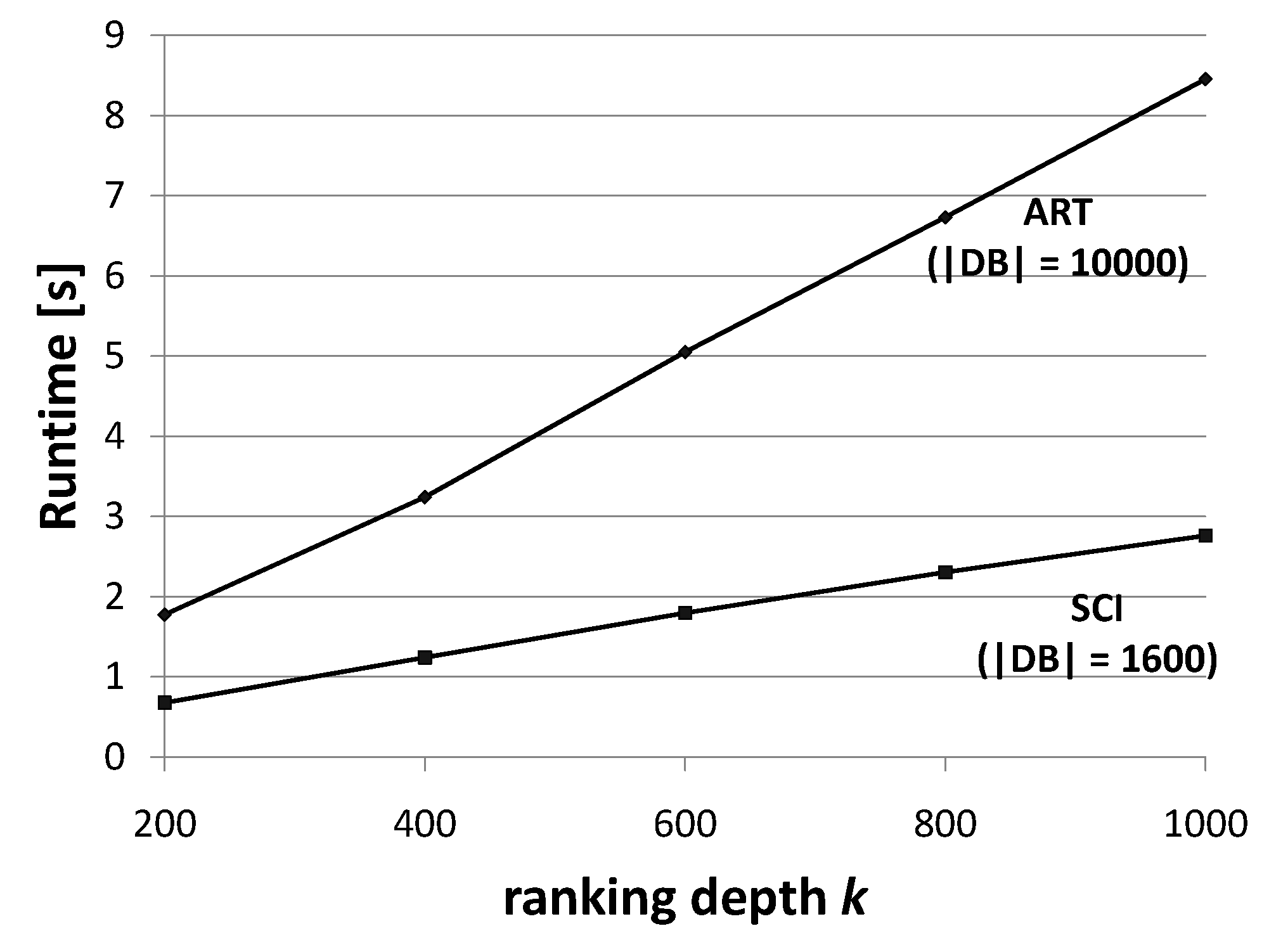}
\caption{Runtime test using \PRaUD\ on the datasets \emph{SCI} and
\emph{ART}.} \label{fig:eval_k_praud}
\end{figure}

\subsection{Ranking Depth $k$}
\label{subsec:rankingDepth}

The influence of the ranking depth $k$ on the runtime performance
of our probabilistic ranking method \PRaUD\ is studied in the next
experiment. As depicted in Figure \ref{fig:eval_k_praud}, where
the experiments were performed using both the \emph{SCI} and the
\emph{ART} dataset, the influence of an increasing $k$ yields a
linear effect on the runtime of \PRaUD, but does not depend on the
type of the dataset. This effect can be explained by taking into
consideration that each iteration of Case 2 or Case 3 requires a
probability computation for each ranking position $0\leq i \leq
k$.

\subsection{Influence of the Degree of Uncertainty}
\label{subsec:uncertainty}

\begin{figure*}[ht]
\centering \hspace{2em}\subfigure[Evaluation of \PRaUD\ by an
increasing uncertainty degree.]{
\includegraphics[width=0.4\columnwidth]{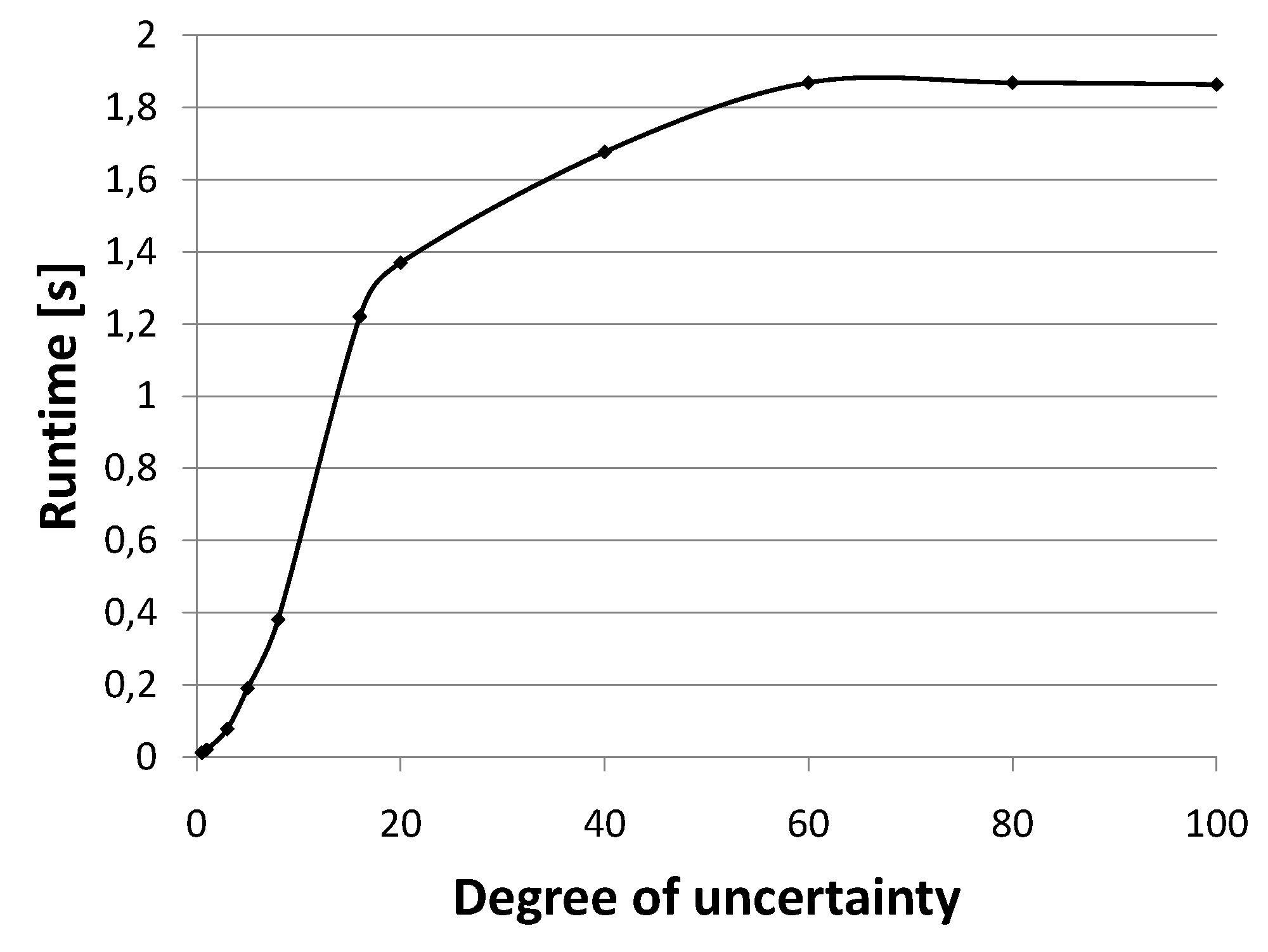}
\label{fig:eval_variance_praud} } \hspace{2em}\subfigure[\Feifei\
vs. \PRaUD\ in a logarithmic scale w.r.t. different
$\varnothing(|AOL|)$ values.]{
\includegraphics[width=0.4\columnwidth]{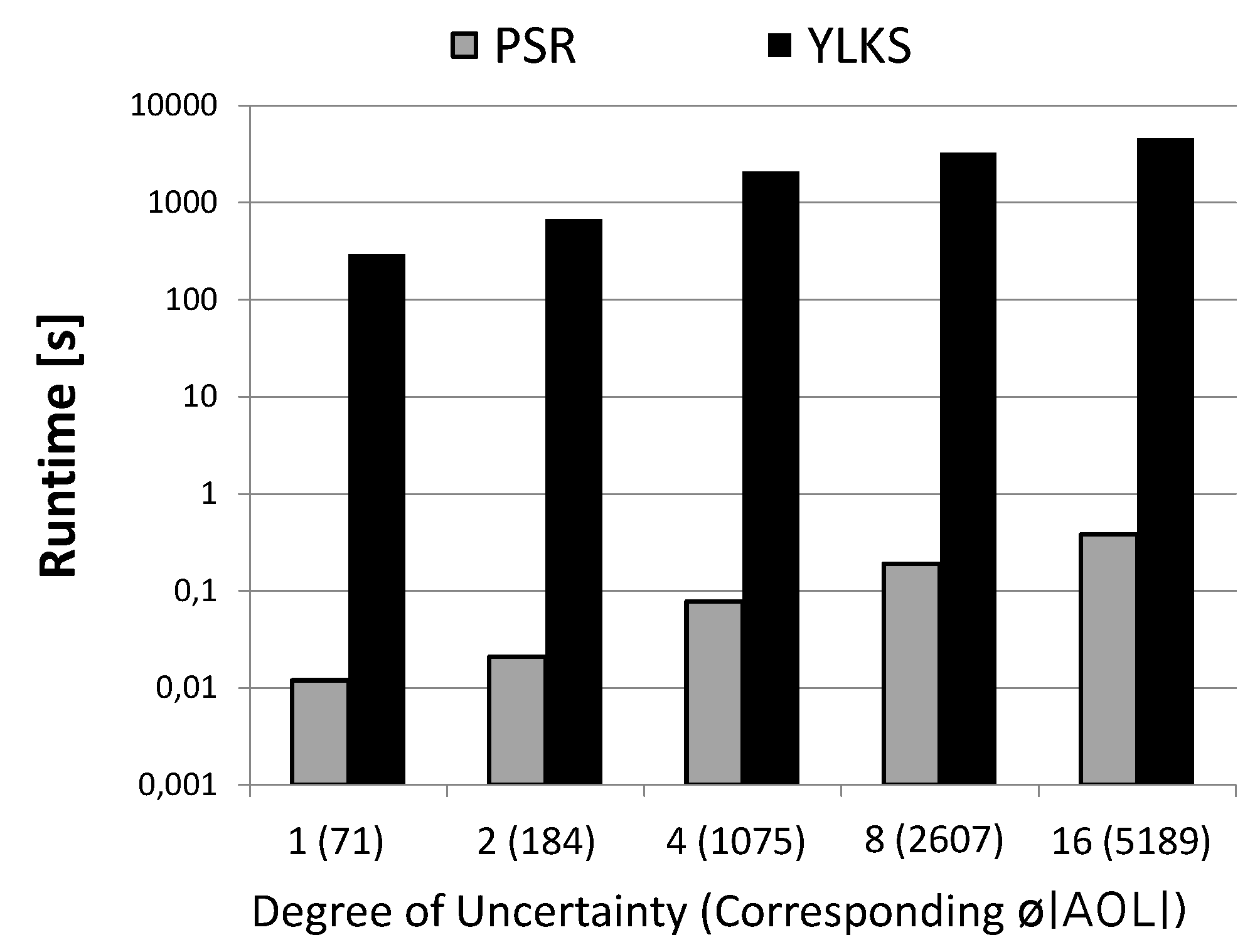}
\label{fig:eval_aol_feifei_praud} } \hspace{\fill}\caption{Runtime
test w.r.t. the degree of uncertainty.} \label{fig:eval_variance}
\end{figure*}

In the next experiment, we varied the uncertainty degree of
objects using the \emph{ART\_2} dataset. In the following
experiments, the ranking depth is set to a fixed value of $k =
100$. As previously discussed, a varying degree of uncertainty
leads to an increase of the overlap between the instances of the
objects and thus, objects will remain in the $AOL$ for a longer
time. The influence of the degree of uncertainty depends on the
probabilistic ranking algorithm. This statement is underlined by
the experiments shown in Figure \ref{fig:eval_variance}. It can be
seen in Figure \ref{fig:eval_variance_praud} that \PRaUD\ scales
superlinear in the degree of uncertainty until a maximal value is
reached. This maximal value is reached, when the degree of
uncertainty approximates a uniform distribution on the whole
vector space for all objects. With an increasing object
uncertainty, the average number of active objects contained in the
\emph{AOL} grows, because the increased overlap of the object
instances causes objects to stay in the \emph{AOL} for a longer
duration. A comparison of the runtime of \Feifei\ and \PRaUD\
w.r.t. the average \emph{AOL} size is depicted in Figure
\ref{fig:eval_aol_feifei_praud}.

\subsection{Summary}
\label{subsec:summary_experiments}

The experiments presented in this section show that the
theoretical analysis of our approach given in Section
\ref{sec:ProbabilisticRankingUncertainObjects} can be confirmed
empirically on both artificial and real-world data. The
performance studies showed that our framework computing the rank
probabilities indeed reduces the quadratic runtime complexity of
state-of-the-art approaches to linear. Note that the cost required
to pre-sort the object instances are neglected in our settings. It
could be shown that our approach scales very well even for large
databases. The speed-up gain of our approach w.r.t. the rank depth
$k$ has shown to be constant, which proofs that both approaches
scale linear in $k$. Furthermore, we could observe that our
approach is applicable for databases with a high degree of
uncertainty (i.e. the degree of variance of the instance
distribution).

\section{Conclusions}
\label{sec:conclusions} In this paper, we proposed a framework for
efficient computation of probabilistic similarity ranking queries
in uncertain vector databases. We introduced a novel concept that
achieves a log-linear runtime complexity in contrast to the
best-known existing approach that solve the same problem with
quadratic runtime complexity. Our concepts are theoretically and
empirically proved to be superior to all existing approaches. In
an experimental evaluation, we showed that our approach scales
very well and, thus, is applicable even for large databases. As
future work, we plan to extend the concepts proposed in this paper
to further uncertainty models.

\bibliographystyle{abbrv}
\bibliography{abbrevs,uncertain}

\end{document}